\pdfoutput=1
\documentclass[11pt]{article}
\usepackage{microtype}
\usepackage[backref=page]{hyperref}
\usepackage{fullpage}
\usepackage{amsthm,amssymb,amsmath}
\usepackage[linesnumbered,ruled,vlined]{algorithm2e}
\usepackage[noend]{algpseudocode}
\usepackage{thm-restate}
\usepackage{color}
\usepackage{tikz}
\usepackage{subcaption}
\allowdisplaybreaks

\SetKwProg{Procedure}{Procedure}{}{}
\SetKwInput{Input}{input}

\newtheorem{theorem}{Theorem}[section]
\newtheorem{lemma}[theorem]{Lemma}
\newtheorem{claim}[theorem]{Claim}
\newtheorem{definition}[theorem]{Definition}

\newcommand{\E}{\mathop{\mathbf{E}}}

\newenvironment{proofof}[1]{\begin{trivlist} \item {\bf Proof
#1:~~}}
  {\qed\end{trivlist}}

\title{Sensitivity Analysis of the Maximum Matching Problem}
\author{Yuichi Yoshida\thanks{Supported by JST, PRESTO Grant Number JPMJPR192B, Japan.}\\
National Institute of Informatics\\
JST, PRESTO\\
\texttt{yyoshida@nii.ac.jp}
\and
Samson Zhou\thanks{Supported by a Simons Investigator Award of David P. Woodruff.}\\
Carnegie Mellon University\\
\texttt{samsonzhou@gmail.com}
}
\date{\today}
\begin{document}

\maketitle
\begin{abstract}
  We consider the \emph{sensitivity} of algorithms for the maximum matching problem against edge and vertex modifications.
  When an algorithm $A$ for the maximum matching problem is deterministic, the sensitivity of $A$ on $G$ is defined as $\max_{e \in E(G)}|A(G) \triangle A(G - e)|$, where $G-e$ is the graph obtained from $G$ by removing an edge $e \in E(G)$ and $\triangle$ denotes the symmetric difference.
  When $A$ is randomized, the sensitivity is defined as $\max_{e \in E(G)}d_{\mathrm{EM}}(A(G),A(G-e))$, where $d_{\mathrm{EM}}(\cdot,\cdot)$ denotes the earth mover's distance between two distributions.
	Thus the sensitivity measures the difference between the output of an algorithm after the input is slightly perturbed.
  Algorithms with low sensitivity, or \emph{stable} algorithms are desirable because they are robust to edge failure or attack.

  In this work, we show a randomized $(1-\epsilon)$-approximation algorithm with \emph{worst-case} sensitivity $O_{\epsilon}(1)$, which substantially improves upon the $(1-\epsilon)$-approximation algorithm of Varma and Yoshida (arXiv 2020) that obtains \emph{average} sensitivity $n^{O(1/(1+\epsilon^2))}$ sensitivity algorithm, and show a deterministic $1/2$-approximation algorithm with sensitivity $\exp(O(\log^*n))$ for bounded-degree graphs.
  We then show that any deterministic constant-factor approximation algorithm must have sensitivity $\Omega(\log^* n)$.
  Our results imply that randomized algorithms are strictly more powerful than deterministic ones in that the former can achieve sensitivity independent of $n$ whereas the latter cannot.
  We also show analogous results for vertex sensitivity, where we remove a vertex instead of an edge.
	As an application of our results, we give an algorithm for the online maximum matching with $O_{\epsilon}(n)$ total replacements in the vertex-arrival model.
	By comparison, Bernstein~et~al.~(J. ACM 2019) gave an online algorithm that always outputs the maximum matching, but only for bipartite graphs and with $O(n\log n)$ total replacements.

	Finally, we introduce the notion of normalized weighted sensitivity, a natural generalization of sensitivity that accounts for the weights of deleted edges.
	For a graph with weight function $w$, the normalized weighted sensitivity is defined to be the sum of the weighted edges in the symmetric difference of the algorithm normalized by the altered edge, i.e., $\max_{e \in E(G)}\frac{1}{w(e)}w\left(A(G) \triangle A(G - e)\right)$.
	Hence the normalized weighted sensitivity measures the weighted difference between the output of an algorithm after the input is slightly perturbed, normalized by the weight of the perturbation.
	We show that if all edges in a graph have polynomially bounded weight, then given a trade-off parameter $\alpha>2$, there exists an algorithm that outputs a $\frac{1}{4\alpha}$-approximation to the maximum weighted matching in $O(m\log_{\alpha} n)$ time, with normalized weighted sensitivity $O(1)$.
\end{abstract}

\thispagestyle{empty}
\setcounter{page}{0}
\newpage


\section{Introduction}
The problem of finding the maximum matching in a graph is a fundamental problem in graph theory with a wide range of applications in computer science.
For example, the maximum matching problem on a bipartite graph $G$ captures a typical example where a number of possible clients want to access content distributed across multiple providers.
Each client can download their specific content from a specific subset of the possible providers, but each provider can only connect to a limited number of clients.
A maximum matching between clients and providers would ensure that the largest possible number of clients receive their content.

However in many modern applications, the underlying graph $G$ represents some large dataset that is often dynamic or incomplete.
In the above example, the content preference of clients may change, which alters the set of suppliers that provide their desired content.
Connections between specific providers and clients may become online or offline, effectively adding or removing edges in the underlying graph.
Providers and clients may themselves join or leave the network, adding or removing entire vertices from the graph.
Thus, it is reasonable to assume that our knowledge of important properties of $G$ may also change or be incomplete.
Nevertheless, we must extract information from our current knowledge of $G$ either for pre-processing or to perform tasks on the current infrastructure.
At the same time, we would like to maintain as much consistency as possible when updates to $G$ are revealed.

Motivated by a formal definition of consistency of algorithms across graph updates, Varma and Yoshida~\cite{VY19:sensitivity} first defined the \emph{average sensitivity} of a deterministic algorithm $A$ to be the Hamming distance\footnote{Here we regard the output as a binary string so we can think of the Hamming distance between outputs.} between the output of $A$ on graphs $G$ and $G-e$, where $G'$ is the graph formed by deleting a random edge of $G$.
Then, they defined \emph{average sensitivity} for randomized algorithms as
\[\underset{e\sim E(G)}{\mathbb{E}}\left[d_{\mathrm{EM}}(A(G),A(G-e))\right],\]
where $d_{\mathrm{EM}}(\cdot,\cdot)$ denotes the earth mover's distance and $G-e$ is the graph obtained from $G$ by deleting an edge $e \in E(G)$.
For the maximum matching problem, they showed a randomized $1/2$-approximation algorithm with average sensitivity $O(1)$ and a randomized $(1-\epsilon)$-approximation algorithm with average sensitivity $O(n^{1/(1+\epsilon^2)})$.

\paragraph{Worst case sensitivity.}
In this work, we continue the study of sensitivity for the maximum matching problem.
Instead of average sensitivity as in~\cite{VY19:sensitivity}, we consider a stronger notion of (worst-case) sensitivity.
Specifically, the \emph{sensitivity} of a deterministic algorithm $A$ is the maximum Hamming distance between the output of $A$ on graphs $G$ and $G'$, where $G'$ is the graph formed by deleting an edge of $G$.
Then, the \emph{sensitivity} of a randomized algorithm $A$ is
\[
  \max_{e \in E(G)}d_{\mathrm{EM}}(A(G),A(G-e)).
\]
Clearly, the sensitivity of an algorithm is no smaller than its average sensitivity.
As a natural variant, we also consider \emph{vertex sensitivity}, where we delete a vertex instead of an edge.
To avoid confusion, sensitivity with respect edge deletion will be sometimes called \emph{edge sensitivity}.

\subsection{Our Contributions}
We first show that, for any $\epsilon>0$, there exists a randomized $(1-\epsilon)$-approximation algorithm whose sensitivity solely depends on $\epsilon$ (Section~\ref{sec:randomized}).
\begin{restatable}{theorem}{thmrndmatching}\label{thm:matching}
  For any $\epsilon > 0$, there exists an algorithm that outputs a $(1-\epsilon)$-approximation to the maximum matching problem with probability at least $0.99$, using time complexity $O((n+m)\cdot K)$ and edge/vertex sensitivity $O(3^K)$, where $K = {(1/\epsilon)}^{2^{O(1/\epsilon)}}$.
\end{restatable}
This result improves upon the previous $(1-\epsilon)$-approximation algorithm~\cite{VY19:sensitivity} in that (1) the sensitivity is constant instead of $O(n^{1/(1+\epsilon^2)})$ and (2) it bounds worst-case sensitivity instead of average sensitivity.

We observe that approximation is necessary to achieve a small sensitivity.
For example, consider an $n$-cycle for an even $n$, and let $M_1$ and $M_2$ be the two maximum matchings of size $n/2$ in the graph.
Consider a deterministic algorithm that always outputs a maximum matching, say, $M_1$ for the $n$-cycle.
Then, it must output $M_2$ after removing an edge in $M_1$, and hence the sensitivity is $\Omega(n)$.
With a similar reasoning, we can show a lower bound of $\Omega(n)$ for randomized algorithms.
Also as we show in Section~\ref{subsec:lb-randomized}, the dependency on $\epsilon$ in Theorem~\ref{thm:matching} is necessary.

One application of our low-sensitivity maximum matching algorithm is the online maximum matching problem with replacements, where updates to the graph $G$ arrive sequentially as a data stream and at all times over the stream, the algorithm must output a matching that is a ``good'' approximation to the maximum matching.
The number of replacements at each time is informally the number of edges in the output matching that differ from the previous output matching, and the goal is to minimize the total number of replacements across the duration of the algorithm.
\begin{restatable}{theorem}{thmmmrecourse}\label{thm:mm:recourse}
There exists an online algorithm that outputs a $(1-\epsilon)$-approximation to the online maximum matching problem with probability $0.99$ and has $O_{\epsilon}(n)$ total replacements.
\end{restatable}
By comparison, Bernstein~et~al.~\cite{BernsteinHR19} gave an online algorithm that always outputs the maximum matching, but has $O(n\log n)$ total replacements and is only restricted to bipartite graphs.
Thus our algorithm achieves worse approximation guarantees than the algorithm of~\cite{BernsteinHR19}, but better total number of replacements and applies for general graphs, rather than only for bipartite graphs,

Next, we show a deterministic algorithm for finding a maximal matching on bounded degree graphs that has low sensitivity (Section~\ref{sec:deterministic}).
Note that it has approximation ratio $1/2$ because the size of any maximal matching is a $1/2$-approximation to the maximum matching.
\begin{restatable}{theorem}{thmdetmatching}\label{thm:det:matching}
There exists a deterministic algorithm that finds a maximal matching with edge/vertex sensitivity $\Delta^{O\left(6^\Delta+\log^* n\right)}$, where $\Delta$ is the maximum degree of a vertex in the graph.
\end{restatable}

Then, we show that randomness is necessary to achieve sensitivity independent of $n$ (Section~\ref{sec:lb}):
\begin{restatable}{theorem}{thmlbdeterministic}\label{thm:lb-deterministic}
    Any deterministic constant-factor approximation algorithm for the maximum matching problem has edge sensitivity $\Omega(\log^* n)$.
\end{restatable}
Namely, we show in Section~\ref{subsec:lb-deterministic-greedy} that we cannot obtain sublinear sensitivity just by derandomizing the randomized greedy algorithm.
Theorems~\ref{thm:matching} and~\ref{thm:lb-deterministic} imply that randomized algorithms are strictly more powerful than deterministic ones in that the former can achieve sensitivity independent of $n$ whereas the latter cannot for the maximum matching problem.

We then introduce the idea of \emph{weighted sensitivity}, which is a natural generalization of sensitivity for both the average and worst cases.
For the problems that we consider, the sensitivity of a deterministic graph algorithm is the number of edges that changes in the output induced by the alteration of a single vertex/edge in the input.
Thus for a weighted graph, the weighted sensitivity is the total weight of the edges that are changed in the output, following the deletion of a vertex/edge.
For randomized algorithms, the definition extends naturally to the earth mover's distance between the distributions with the corresponding weighted loss function.
Finally, we can also normalize by the weight of the edge that is deleted.

The motivation for studying weighted sensitivity is natural; in many applications with evolving data, the notion of sensitivity arises in the context of recourse, a quantity that measures the change in the underlying topology of the optimal solution.
For example in the facility location problem, the goal is to construct a set of facilities to minimize the sum of the costs of construction and service to a set of consumers.
As the information about the set of consumers evolves, it would be ideal to minimize the number of relocations for the facilities, due to the construction costs, which is measured by the sensitivity of the algorithm.
However, as construction costs may not be uniform, a more appropriate quantity to minimize would be the total cost of the relocations for the facilities, which is measured by the weighted sensitivity.

Similarly, matchings are often used to maximize flow across a bipartite graph, but the physical structures that support the flow may incur varying costs to construct or demolish, corresponding to the amount of flow that the structures support.
In this case, we note that it may not be possible for the worst case weighted sensitivity to be small.
For example, if a single edge has weight $n^C$ for some large constant $C$ and the remaining edges have weight $1$, any constant factor approximation to the maximum weighted matching must include the heavy edge.
But if the heavy edge is then removed from the graph, the weighted sensitivity of any constant factor approximation algorithm is $\Omega(n^C)$.
This issue is circumvented by the normalized weighted sensitivity, which scales the sensitivity by the weight of the deleted edge.
We give approximation algorithms for maximum weighted matching with low normalized weighted worst-case sensitivity.

\begin{theorem}
Let $G=(V,E)$ be a weighted graph with $\frac{1}{n^c}\le w(e)\le n^c$ for some constant $c>0$ and all $e\in E$.
For a trade-off parameter $\alpha>2$, there exists an algorithm that outputs a $\frac{1}{4\alpha}$-approximation to the maximum weighted matching in $O(m\log_{\alpha} n)$ time and has normalized weighted sensitivity $O(1)$.
\end{theorem}

Our results also extend to $\alpha=2$ and general worst-case weighted sensitivity, i.e., weighted sensitivity that is not normalized.
We detail these algorithms in Section~\ref{sec:weighted}.

\subsection{Proof Sketch}
We explain the idea behind the algorithm of Theorem~\ref{thm:matching}.
For simplicity, we focus on edge sensitivity.
We note that if we only sought a $1/2$-approximation to the maximum matching, then it would suffice to find any maximal matching.
Although the well-known greedy algorithm produces a maximal matching, the output of the algorithm is highly sensitive to the ordering of the edges in the input.
One may hope that, if we choose an ordering of the edges uniformly at random, then the resulting output will be stable against edge deletions to the underlying graph.
This is not immediately obvious because the deleted edge will appear about halfway through the ordering (of the edges in the original graph) in expectation, so it seems possible that it can impact about the remaining half of the edges.
Luckily, we show that the edges at the beginning of the ordering are significantly more important, so that even if the deleted edge appears about halfway through the ordering, the sensitivity of the maximal matching is $O(1)$ (Section~\ref{sec:greedy}).
Our analysis is similar to~\cite{censor2016optimal}, who show that the vertices at the beginning of an ordering are significantly more important in maintaining a maximal independent set in the dynamic distributed model.

Adapting this idea to a $(1-\epsilon)$-approximation is more challenging.
The natural approach is to take a maximal matching and repeatedly find a large number of augmenting paths, but the change of even a single edge in a maximal matching can potentially impact a large number of edges if the augmenting paths are found in a sequential manner.
We instead adapt a layered graph of~\cite{mcgregor2005finding} that is used to randomly find a large number of augmenting paths in a small number of passes in the streaming model.
Crucially, we instead find a large number of disjoint augmenting paths in a small number of parallel rounds, which results in low sensitivity.

Now we turn to explaining the idea behind Theorem~\ref{thm:det:matching}.
Again we focus on edge sensitivity.
Our algorithm first uses a deterministic local computation algorithm (LCA) of~\cite{ColeV86} for $6^{\Delta}$-coloring a graph $G$ with maximum degree $\Delta$, using $O(\Delta\log^*n)$ probes to an adjacency list oracle.
Here we want to design an algorithm that answer queries about the colors of vertices by making a series of probes to the oracle.
The answers of the algorithm must be consistent so that there exists at least one proper coloring that is consistent with the answers.
In our case, each probe to the oracle is a query $(v,i)$ with $v\in V$ and a positive integer $i$.
If the degree of $v$ is at least $i$, the oracle responds with the $i$-th neighbor of $v$ to the probe.
Otherwise, the oracle outputs a special symbol $\bot$.
In particular, the deterministic $6^{\Delta}$-coloring LCA only probes vertices that are within a ``small'' neighborhood of the query.

Given a coloring for $G$, we then give a local distributed algorithm that takes a coloring of a graph and outputs a maximal matching.
It follows from a framework of~\cite{ParnasR07} that our local distributed algorithm can actually be simulated by a deterministic LCA that again only probes a ``small'' neighborhood of the query.
Thus to bound the sensitivity of the algorithm, we bound the number of queries for which a deleted edge would be probed.
Since only a small number of queries probes the deleted edge, then the output of the algorithm only has a small number of changes and thus low worst-case sensitivity.

Our lower bound of Theorem~\ref{thm:lb-deterministic} considers the set of length-$t$ cycles on a graph with $n$ vertices.
Any matching on length-$t$ cycles can be represented as a series of indicator variables denoting whether edge $i\in[t]$ is in the matching.
We can then interpret the indicator variables as an integer encoding from $0$ to $2^t-1$ through the natural binary representation.
Ramsey theory claims that for $t=O(\log^* n)$, there exists a set $S$ of $t+1$ nodes of $n$ so that any subset of $t$ nodes has the same encoding.
We then choose $G$ and $G'$ to be the cycle graphs consisting of the first $t$ nodes of $S$ and the last $t$ nodes of $S$, respectively.
Since the encodings of the matchings of $G$ and $G'$ are the same, but the edge indices are shifted by one, it follows that $\Omega(t)$ edges must be in the symmetric difference between $G$ and $G'$, which implies from $t=O(\log^* n)$ that the worst-case sensitivity of the algorithm must be $\Omega(\log^* n)$.

\subsection{Related Work}
Varma and Yoshida~\cite{VY19:sensitivity} introduced the notion of sensitivity and performed a systematic study of average sensitivity on many graph problems.
Namely, they gave efficient approximation algorithms with low average sensitivities for the minimum spanning forest problem, the global minimum cut problem, the minimum $s$-$t$ cut problem, and the maximum matching problem.
They also introduced a low-sensitivity algorithm for linear programming, and proved many fundamental properties of average sensitivity, such as sequential or parallel composition.
Peng and Yoshida~\cite{PengY20} gave an algorithm for the problem of spectral clustering with average sensitivity $\frac{\lambda_2}{\lambda_3^2}$, where $\lambda_i$ is the $i$-th smallest eigenvalue of the normalized Laplacian, which is small when there are exactly two clusters in the graph.

The effects of graph updates have also been studied significantly in the dynamic/online model, where updates to the graph arrive in a stream, and the goal is to maintain some data structure to answer queries on the underlying graph so that both the update time and query time are efficient.
Consequently, most of the literature for dynamic algorithms focuses on optimizing these quantities, rather than the changes in the output as the data evolves.
Sensitivity analysis is more relevant when the goal of the dynamic/offline model is to minimize the number of changes between successive outputs of the algorithm over the stream.

Lattanzi and Vassilvitski~\cite{LattanziV17} studied the problem of consistent $k$-clustering, where the goal is to maintain a constant-factor approximation to some underlying $k$-clustering problem, such as $k$-center, $k$-median, or $k$-means, while minimizing the total number of changes to the set of centers as the stream evolves.
In this setting, each change to the set of center is known as a \emph{recourse}.
Whereas the model of~\cite{LattanziV17} allows only insertions of new points, algorithms with low sensitivity are robust against both insertions and deletions.
Cohen-Addad~\emph{et.~al.}~\cite{Cohen-AddadHPSS19} further considered the facility location problem in this model of maintaining a constant-factor approximation while minimizing the total recourse.
Although the algorithm of~\cite{Cohen-AddadHPSS19} addresses both the insertions and deletions of points, their total recourse across the stream is $O(n)$, where $n$ is the length of the stream; this is inherent to the difficulty of their problem in the model.
Whereas their work already provides an amortized $O(1)$ recourse per update, we also study the worst-case sensitivity in our work.

Consistency for maximum matching has also been thoroughly studied, called the online matching problem with replacements.
The problem was introduced by Grove~et~al.~\cite{GroveKKV95} for bipartite graphs, who gave matching upper and lower bounds of $\Theta(n\log n)$ total replacements when all vertices on one side of the partition have degree two.
Chaudhuri~et~al.~\cite{ChaudhuriDKL09} showed that the greedy algorithm that repeatedly adds the shortest augmenting path from the newest arrived vertex has $\Theta(n\log n)$ total replacements in expectation for any arbitrary underlying bipartite graph, provided that the vertices on one side of the partition arrive in a random order.
They also gave an algorithm with $O(n\log n)$ total replacements for acyclic bipartite graphs, as well as a tight asymptotic lower bound.
For general bipartite graphs, Bosek~et~al.~\cite{BosekLSZ14} showed an algorithm with $O(n\sqrt{n})$ total replacements, using total time $O(m\sqrt{n})$, matching the best offline maximum matching algorithm for static bipartite graphs.
Recently, Bernstein~et~al.~\cite{BernsteinHR19} gave an algorithm for online maximum bipartite matching with $O(n\log^2 n)$ total replacements, substantially progressing toward the strongest known lower bound, which is $\Omega(n\log n)$~\cite{GroveKKV95}.

\subsection{Preliminaries}
For a positive integer $n$, let $[n]$ denote the set $\{1,2,\ldots,n\}$.
For a positive integer $n$ and $p \in [0,1]$, let $\mathcal{B}(n,p)$ be the binomial distribution with $n$ trials and success probability $p$.
We use the notation $O_\epsilon(\cdot)$ to omit dependencies on $\epsilon$. 

Let $G=(V,E)$ be a graph.
For an edge $e \in E$, let $N_G(e)$ be the ``neighboring'' edges of $e$ in $G$, that is, $N_G(e) = \{e' \in E \mid e' \neq e, |e' \cap e| \geq 1 \}$.
We omit the subscript if it is clear from the context.

For two (vertex or edge) sets $S$ and $S'$, let $d_{\mathrm{H}}(S,S') = |S \triangle S'|$, where $\triangle$ denotes the symmetric difference.
Abusing the notation, for set of paths $\mathcal{P}$ and $\mathcal{P}'$, we write $d_{\mathrm{H}}(\mathcal{P},\mathcal{P}')$ to denote $d_{\mathrm{H}}(\cup_{P \in \mathcal{P}} V(P), \cup_{P \in \mathcal{P}'}V(P))$.
For two random sets $X$ and $X'$, let $d_{\mathrm{EM}}(X,X')$ be the earth mover's distance between $X$ and $X$, where the distance between two sets is measured by $d_{\mathrm{H}}$, that is,
\[
  d_{\mathrm{EM}}(X,X') = \min_{\mathcal{D}} \E_{(S,S') \sim \mathcal{D}}d_{\mathrm{H}}(S,S'),
\]
where $\mathcal{D}$ is a distribution such that its marginal on the first and second coordinates are $X$ and $X'$, respectively.
For a real-valued function $\beta$ on graphs, we say that the \emph{sensitivity} of a (randomized) algorithm $A$ that outputs a set of edges is at most $\beta$ if
\[
  d_{\mathrm{EM}}(A(G),A(G-e)) \leq \beta(G)
\]
holds for every $e \in E(G)$.

Given a matching $M$ in a graph $G = (V, E)$, we call a vertex \emph{free} if it does not appear as the endpoint of any edge in $M$.
A path $(v_1, v_2 ,\ldots,v_{2\ell+2})$ of length $2\ell+1$ is an \emph{augmenting path} if $v_1$ and $v_{2\ell+2}$ are free vertices and $(v_i,v_{i+1}) \in M$ for even $i$ and $(v_i, v_{i+1}) \in E \setminus M$ for odd $i$.

\section{Randomized \texorpdfstring{$(1-\epsilon)$}{(1-epsilon)}-Approximation}\label{sec:randomized}
In this section, we prove Theorem~\ref{thm:matching}.
Our algorithm, which we describe in Section~\ref{subsec:algorithm-description}, is a slight modification of the multi-pass streaming algorithm due to McGregor~\cite{mcgregor2005finding}.
We discuss its approximation guarantee and sensitivity in Sections~\ref{subsec:randomized-approximation-ratio} and~\ref{subsec:randomized-sensitivity}, respectively.
Finally, we discuss applications to online matching with replacements in Section~\ref{subsec:online-matching}.

\subsection{Algorithm Description}\label{subsec:algorithm-description}

A key step of McGregor's algorithm is to find a large set of augmenting paths of a specified length in a batch manner using the \emph{layered graph}, given below.
Given a graph $G=(V,E)$, a matching $M \subseteq E$, and a positive integer $\ell$, the layered graph $H = H(G)$ consists of $\ell+2$ layers $L_0,L_1,\ldots,L_{\ell+1}$, where $L_0 = L_{\ell+1} = V$ and $L_1 = L_2 = \cdots = L_\ell = V \times V$.

For each vertex $v \in V$, we sample $i_v \in \{0,\ell+1\}$ uniformly at random independently from others.
We say that the copy of $v$ in the $L_{i_v}$-th layer is \emph{active} and that the other copy is \emph{inactive}.
For each edge $\{u,v\}\in M$, with probability half, we sample a value $i_{(u,v)} \in \{1,\ldots,\ell\}$ uniformly at random and set $i_{(v,u)} = \bot$, where $\bot$ is a special symbol, and with the remaining probability half, we sample a value $i_{(v,u)} \in \{1,\ldots,\ell\}$ uniformly at random and set $i_{(u,v)} = \bot$.
For each edge $\{u,v\} \in E \setminus M$, we set $i_{(u,v)} = i_{(v,u)} = \bot$.
We say that the copy of $(u,v)$ in the $L_{i_{(u,v)}}$-th is \emph{active} if $i_{(u,v)} \neq \bot$ and is \emph{inactive} otherwise.
Intuitively, some orientation of each edge $\{u,v\}$ in the matching $M$ is assigned to a random internal layer in $H$ and edges of $G$ that are not in the matching are not initially assigned to any layer in $H$.
For $i = 0,\ldots,\ell+1$, we denote by $\tilde{L}_i$ the set of active vertices in $L_i$.
Let $L = \bigcup_{i=0}^{\ell+1}L_i$ be the vertex set of $H$, and let $\tilde{L} = \bigcup_{i=0}^{\ell+1}\tilde{L}_i$ be the set of active vertices in $H$.

The edges in the layered graph $H$ are those between active vertices that can be a part of an augmenting path in $G$.
More specifically,
\begin{itemize}
  \itemsep=0pt
\item We add an edge between $t \in \tilde{L}_0$ and $(u,v) \in \tilde{L}_1$ if $t$ is free in $M$ and $t$ is adjacent to $v$.
\item We add an edge between $(u,v) \in \tilde{L}_\ell$ and $s \in \tilde{L}_{\ell+1}$ if $s$ is free in $M$ and $s$ is adjacent to $u$.
\item We add an edge between $(u,v) \in \tilde{L}_i$ and $(u',v') \in \tilde{L}_{i+1}$ for $i \in [\ell-1]$ if $v$ is adjacent to $u'$.
\end{itemize}
Note that inactive vertices are isolated in $H$.

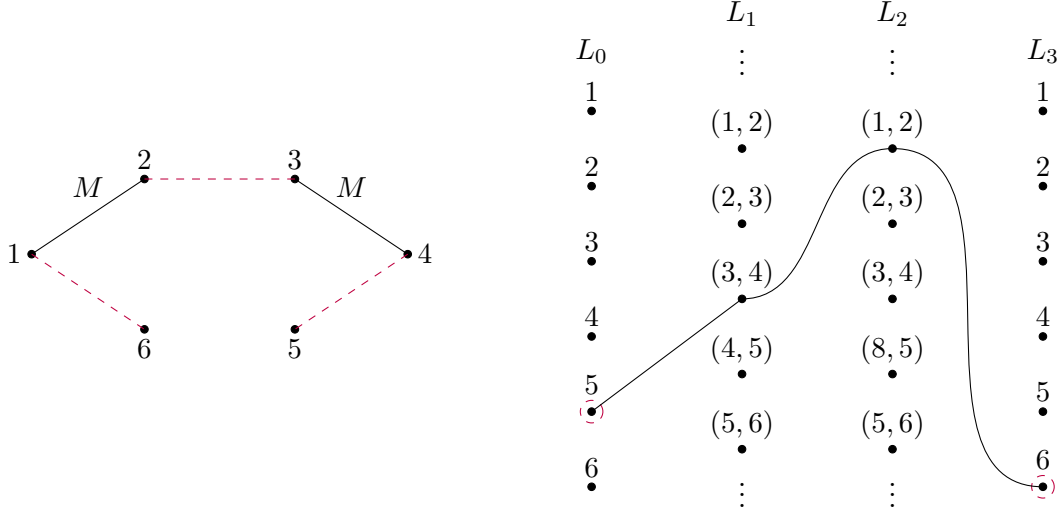
\begin{figure*}[!htb]
\centering
\begin{subfigure}{0.45\textwidth}
\begin{tikzpicture}[scale=2.5]
\node[above] at (0.3,0.25){$M$};
\node[above] at (1.7,0.25){$M$};

\draw[fill] (0,0) circle (0.02);
\node[left] at (0,0){$1$};
\draw[fill] (0.6,0.4) circle (0.02);
\node[above] at (0.6,0.4){$2$};
\draw[fill] (1.4,0.4) circle (0.02);
\node[above] at (1.4,0.4){$3$};
\draw[fill] (2,0) circle (0.02);
\node[right] at (2,0){$4$};
\draw[fill] (1.4,-0.4) circle (0.02);
\node[below] at (1.4,-0.4){$5$};
\draw[fill] (0.6,-0.4) circle (0.02);
\node[below] at (0.6,-0.4){$6$};

\draw (0,0)--(0.6,0.4);
\draw (1.4,0.4)--(2,0);
\draw [purple, dashed] (0,0)--(0.6,-0.4);
\draw [purple, dashed] (0.6,0.4)--(1.4,0.4);
\draw [purple, dashed] (2,0)--(1.4,-0.4);

\end{tikzpicture}
\end{subfigure}
\begin{subfigure}{0.45\textwidth}
\begin{tikzpicture}[scale=0.5]
\node[above] at (0,11){$L_0$};
\draw[fill] (0,0) circle (0.1);
\draw[fill] (0,2) circle (0.1);
\draw[fill] (0,4) circle (0.1);
\draw[fill] (0,6) circle (0.1);
\draw[fill] (0,8) circle (0.1);
\draw[fill] (0,10) circle (0.1);
\node[above] at (0,10){$1$};
\node[above] at (0,8){$2$};
\node[above] at (0,6){$3$};
\node[above] at (0,4){$4$};
\node[above] at (0,2.2){$5$};
\node[above] at (0,0){$6$};
\draw[purple, dashed] (0,2) circle (0.3);

\node[above] at (4,12){$L_1$};
\node at (4,11.5){$\vdots$};
\draw[fill] (4,1) circle (0.1);
\node[above] at (4,1){$(5,6)$};
\draw[fill] (4,3) circle (0.1);
\node[above] at (4,3){$(4,5)$};
\draw[fill] (4,5) circle (0.1);
\node[above] at (4,5){$(3,4)$};
\draw[fill] (4,7) circle (0.1);
\node[above] at (4,7){$(2,3)$};
\draw[fill] (4,9) circle (0.1);
\node[above] at (4,9){$(1,2)$};
\node at (4,0){$\vdots$};

\node[above] at (8,12){$L_2$};
\node at (8,11.5){$\vdots$};
\draw[fill] (8,1) circle (0.1);
\node[above] at (8,1){$(5,6)$};
\draw[fill] (8,3) circle (0.1);
\node[above] at (8,3){$(8,5)$};
\draw[fill] (8,5) circle (0.1);
\node[above] at (8,5){$(3,4)$};
\draw[fill] (8,7) circle (0.1);
\node[above] at (8,7){$(2,3)$};
\draw[fill] (8,9) circle (0.1);
\node[above] at (8,9){$(1,2)$};
\node at (8,0){$\vdots$};

\node[above] at (12,11){$L_3$};
\draw[fill] (12,0) circle (0.1);
\draw[fill] (12,2) circle (0.1);
\draw[fill] (12,4) circle (0.1);
\draw[fill] (12,6) circle (0.1);
\draw[fill] (12,8) circle (0.1);
\draw[fill] (12,10) circle (0.1);
\node[above] at (12,10){$1$};
\node[above] at (12,8){$2$};
\node[above] at (12,6){$3$};
\node[above] at (12,4){$4$};
\node[above] at (12,2){$5$};
\node[above] at (12,0.2){$6$};
\draw[purple, dashed] (12,0) circle (0.3);

\draw (0,2)--(4,5);
\draw (4,5)to[out=0,in=180](8,9);
\draw (8,9)to[out=0,in=180](12,0);
\end{tikzpicture}
\end{subfigure}
\caption{Example of an active layered graph with respect to a matching $M$, in solid lines.
The free vertex $5$ appears in $L_0$ and the free vertex $6$ appears in $L_3$.
The augmenting path found by the layered graph is represented by a dashed purple line.}\label{fig:layered}
\end{figure*}

We introduce the following definition to handle augmenting paths for a matching $M$ in a graph $G$ via paths in its corresponding layered graph.
\begin{definition}
We say that a path $v_i,v_{i-1},\ldots,v_0$ with $v_j \in \tilde{L}_j\;(j \in \{0,1,\ldots,i\})$ is an \emph{$i$-path}.
Note that an $(\ell + 1)$-path in $H$ corresponds to an augmenting path of length $2\ell + 1$ in $G$.
\end{definition}

The layered graph defined above is slightly different from the original one due to McGregor~\cite{mcgregor2005finding} in that he did not include inactive vertices in $H$, as they are irrelevant to find augmenting paths.
However as we consider sensitivity of algorithms, it is convenient to fix the vertex set so that it is independent of the current matching $M$.

We briefly define the randomized greedy subroutine $\textsc{RandomizedGreedy}$ on a graph $G=(V,E)$ as follows.
The subroutine first chooses a random ordering $\pi$ over edges and then starting with an empty matching $M$, the procedure iteratively adds the $i$-th edge in the ordering $\pi$ to $M$ if the edge is not adjacent to any edge in $M$, until it has processed all edges.
See Algorithm~\ref{alg:rand:greedy} for the full details.

\begin{algorithm}[htb!]
  \caption{Randomized Greedy Algorithm}\label{alg:rand:greedy}
  \Procedure{\emph{\Call{RandomizedGreedy}{Graph $G=(V,E)$}}}{
	Generate a permutation $\pi$ of $E$ uniformly at random\;
	Greedily add edges to a maximal matching $M$, in the order of $\pi$\;
	Output $M$\;
	}
\end{algorithm}

Algorithm~\ref{alg:augmenting-paths} shows our algorithm for finding a large set of augmenting paths of length $2\ell+1$ given a matching $M$ in a graph $G$.
For a matching $M$ and a vertex $v$ belonging to an edge $e$ in $M$, let $\Gamma_M(v)$ denote the other endpoint of $e$.
Similarly, for a vertex set $S$ such that each edge in $M$ uses at most one vertex in $S$, let $\Gamma_M(S)$ denote the set of other endpoints.
For subsets $L$ and $R$ of adjacent layers in $H$, let $\Call{RandomizedGreedy}{L,R}$ denote the randomized greedy on the induced bipartite graph $H[L \cup R]$.
\Call{FindPaths}{} tries to find a large set of vertex disjoint $i$-paths from $S \subseteq L_i$ to $L_0$.
The result is stored as a tag function $t: V(H) \to V(H) \cup \{\mathsf{untagged}, \mathsf{dead\ end}\}$.
Here, $t(v)$ is initialized to $\mathsf{untagged}$, and it will represent the next vertex in the $i$-path found.
If we could not find any $i$-path starting from $v$, $t(v)$ is set to $\mathsf{dead end}$.

The difference from McGregor's algorithm is that we run the loop in \Call{FindPaths}{} $1/\delta$ times instead of running it until $|M'| \leq \delta |M|$.
This makes sure that we compute a maximal matching the same number of times no matter what $G$ and $M$ are, and it is more convenient when analyzing the sensitivity.

Our algorithm for the maximum matching problem (Algorithm~\ref{alg:randomized-matching}) simply runs \Call{AugmentingPaths}{} sufficiently many times for various choice of $\ell$ and then keep applying the obtained augmenting paths.
Before analyzing its approximation ratio and sensitivity, we analyze its running time.

\begin{algorithm}[t!]
  \caption{Augmentation Algorithm}\label{alg:augmenting-paths}
  \Procedure{\emph{\Call{AugmentingPaths}{$G,M,\ell,\delta$}}}{
    Construct a layered graph $H$ using $G$, $M$, and $\ell$\;
    $t(v) \leftarrow \mathsf{untagged}$ for every vertex $v$ in $H$\;
    $\Call{FindPaths}{H, L_{\ell+1}, \ell, \delta, t}$\;
    Convert $t$ to a set of augmenting paths in $G$\;
    Apply the augmenting paths to $M$.
  }

  \Procedure{\emph{\Call{FindPaths}{$H,S,i,\delta, t$}}}{
    $M' \leftarrow \Call{RandomizedGreedy}{S, L_{i-1} \cap t^{-1}(\mathsf{untagged})}$\;
    $S' \leftarrow \Gamma_{M'}(S)$\;
    \If{$i=1$}{
      \For{$u \in S$}{
        \If{$u \in \Gamma_{M'}(L_0)$}{$t(u) \leftarrow \Gamma_{M'}(u)$}
        \If{$u \in S \setminus \Gamma_{M'}(L_0)$}{$t(u) \leftarrow \mathsf{dead\; end}$.}
      }
      \Return $t$.
    }
    \For{$\lceil 1/\delta \rceil$ times}{
      \Call{FindPaths}{$H,S',i-1,\delta^2,t$}\;
      \For{$v \in S' \setminus t^{-1}(\mathsf{dead\;end})$}{
        $t(\Gamma_{M'}(v)) \leftarrow v$.
      }
      $M' \leftarrow \Call{RandomizedGreedy}{S \cap t^{-1}(\mathsf{untagged}), L_{i-1} \cap t^{-1}(\mathsf{untagged}) }$\;
      $S' \leftarrow \Gamma_{M'}(S \cap t^{-1}(\mathsf{untagged}))$.
    }
    \For{$v \in S \cap t^{-1}(\mathsf{untagged})$}{
      $t(v) \leftarrow \mathsf{dead\; end}$.
    }
    \Return
  }
\end{algorithm}

\begin{algorithm}[t!]
  \caption{Algorithm for maximum matching}\label{alg:randomized-matching}
  \Procedure{\emph{\Call{Matching}{$G,\epsilon$}}}{
    Let $\pi$ be a random ordering of the edges of $G$\;
    Let $M$ be the greedy maximal matching of $G$ induced by $\pi$\label{line:randomized-matching-0}\;
    $k \gets\lceil \epsilon^{-1}+1\rceil$\;
    $r\gets 4k^2(16k+20)(k-1){(2k)}^k$\;
    \For{$\ell=1$ to $k$}{
      \For{$i=1$ to $r$}{
        $M'_{\ell,i} \leftarrow \Call{AugmentingPaths}{G,M,\ell,\frac{1}{r(2k+2)}}$\;
        $M \leftarrow M \oplus M'_{\ell,i}$.\label{line:randomized-matching-update}
      }
    }
    \Return $M$.
  }
\end{algorithm}

\begin{lemma}\label{lem:time-complexity}
  The total running time of Algorithm~\ref{alg:randomized-matching} is $O(n+m)K$, where $K={(1/\epsilon)}^{2^{O(1/\epsilon)}}$.
\end{lemma}
\begin{proof}
Observe that the outer loop of Algorithm~\ref{alg:randomized-matching} runs for $r$ iterations and the inner loop runs for $k$ iterations, where $k=\lceil \epsilon^{-1}+1\rceil$ and $r=4k^2(16k+20)(k-1){(2k)}^k$.
Each inner loop runs an instance of $\textsc{AugmentingPaths}$ with parameters $\ell\le k$ and $\delta=\frac{1}{r(2k+2)}$, which creates a layered graph with $O(\ell)$ layers in $O((n+m)k)$ time, and then calls $\textsc{FindPaths}$.
For each time that $\textsc{FindPaths}$ is called, the value of $\delta$ is squared and the value of $\ell$ is decremented, starting at $\ell=k$ until $\ell=1$.
Thus, the loop in $\textsc{FindPaths}$ is run at most $\frac{1}{\delta}=O\left({(r(2k+2))}^{2^k}\right)$ times and each loop uses time $O(n+m)$.
Hence, the total runtime is $O(n+m)K$, where $K={(1/\epsilon)}^{2^{O(1/\epsilon)}}$.
\end{proof}

\subsection{Approximation Ratio}\label{subsec:randomized-approximation-ratio}
In this section, we analyze the approximation ratio of Algorithm~\ref{alg:randomized-matching}.

\begin{lemma}\cite{mcgregor2005finding}
Suppose \Call{FindPaths}{$\cdot,\cdot,j,\cdot$} is called $\delta^{-2^{i-j+1}+1}$ times in the recursion for \Call{FindPaths}{$H,S,i,\delta$}.
Then at most $2\delta|L_i|$ paths are removed from consideration as being $(i+1)$-paths.
\end{lemma}
Let $\mathcal{L}_i$ be the set of graphs whose vertices are partitioned into $i+2$ layers, $L_0,\ldots,L_{i+1}$ and whose edges are a subset of $\bigcup_{j=1}^{i+1}(L_{j} \times L_{j-1})$.
Then we immediately have the following lemma, analogous to Lemma 2 in~\cite{mcgregor2005finding}:
\begin{lemma}
 For a graph $G\in\mathcal{L}_i$, \Call{FindPaths}{$H,S,i,\delta$} finds at least $(\gamma-\delta)|M|$ of the $(i+1)$-paths among some maximal set of $(i+1)$-paths of size $\gamma|M|$.
\end{lemma}
We require the following structural property relating maximal and maximum matchings through the set of connected components in the symmetric difference.
\begin{lemma}[Lemma 1 in~\cite{mcgregor2005finding}]\label{lem:alpha:bound}
  Let $M$ be a maximal matching and $M^*$ be a maximum matching.
  Let $\mathcal{C}$ be the set of connected components in $M^*\triangle M$.
  Let $\alpha_{\ell}$ be the constant so that $\alpha_\ell|M|$ is the number of connected components in $\mathcal{C}$ with $\ell$ edges from $M$, excluding those with $\ell$ edges from $M^*$.
  If $\max_{\ell \in[k]}\alpha_\ell \le\frac{1}{2k^2(k+1)}$, then $|M|\ge\frac{|M^*|}{1+1/k}$.
\end{lemma}
We also require the following result by \cite{mcgregor2005finding} bounding the number of augmenting paths found by \textsc{AugmentingPaths}.
\begin{lemma}[Theorem 1 in~\cite{mcgregor2005finding}]\label{lem:round:augment}
  If $G$ has $\alpha_\ell|M|$ augmenting paths of length $2\ell+1$, then the number of augmenting paths of length $2\ell+1$ found by \Call{AugmentingPaths}{} is at least $(b_\ell\beta_\ell-\delta)|M|$, where $b_\ell=\frac{1}{2\ell+1}$ and $\beta_\ell\sim \mathcal{B}\left(\alpha_\ell|M|,\frac{1}{2{(2\ell)}^\ell}\right)$.
\end{lemma}

\noindent
We now show that Algorithm~\ref{alg:randomized-matching} outputs a $(1-\epsilon)$-approximation to the maximum matching. 
\begin{theorem}\label{thm:randomized-approximation}
Algorithm~\ref{alg:randomized-matching} finds a $(1-\epsilon)$-approximation to the maximum matching with probability at least $0.99$.
\end{theorem}
\begin{proof}
  We say the algorithm enters \emph{phase} $\ell$ when the number of layers in the layered graph has been incremented to $\ell$, i.e., each invocation of the outer for loop corresponds to a separate phase.
  We say the algorithm enters \emph{round} $i$ in phase $\ell$ after the subroutine \Call{AugmentingPaths}{} has completed $i-1$ iterations within phase $\ell$.
  Let $M_{\ell,i}$ be the matching $M$ prior to the call to \Call{AugmentingPaths}{} in round $i$ of phase $\ell$.
  Let $\alpha_{\ell,i}|M_{\ell,i}|$ be the number of length $2\ell+1$ augmenting paths of $M_{\ell,i}$.
  Thus by Lemma~\ref{lem:round:augment}, the subroutine \Call{AugmentingPaths}{} augments $M_{\ell,i}$ by at least $(b_{\ell}\beta_{\ell,i}-\delta)|M_{\ell,i}|$ edges in round $i$ of phase $\ell$, where $b_{\ell}=\frac{1}{2\ell+2}$, $\delta$ is a parameter that we choose later, and $\beta_{\ell,i}$ is a random variable distributed according to $\mathcal{B}\left(\alpha_{\ell,i}|M_{\ell,i}|,\frac{1}{2{(2\ell)}^\ell}\right)$.
  Let $M$ be the output matching.
  Then by Bernoulli's inequality, we have
  \begin{align*}
  \Pr\left[|M| \ge2|M_{1,1}|\right]
  & \geq \Pr\left[|M_{1,1}|\prod_{\ell\in[k],i\in[r]}\left(1+\max\left(0,b_{\ell}\beta_{\ell,i}-\delta\right)\right)\ge2|M_{1,1}|\right] \\
  & \ge \Pr\left[\sum_{\ell\in[k]}\max_{i\in[r]}b_{\ell}\beta_{\ell,i}\ge2+r\delta\right].
  \end{align*}

	We would like to analyze $\sum_{\ell\in[k]}\max_{i\in[r]}b_{\ell}\beta_{\ell,i}$, but the analysis is challenging due to dependencies between multiple rounds and phases.
	We thus define independent variables $X_1,\ldots,X_k$ and use a coupling argument.

  We define $A_{\ell}=\max_{i \in [r]}b_{\ell}\beta_{\ell,i}|M_{\ell,i}|$ to be an upper bound on the maximum number of augmented edges during phase $\ell$ of the algorithm.
  Suppose by way of contradiction that $\max_{i \in [r]}\alpha_{\ell,i}<\alpha_0:=\frac{1}{2k^2(k-1)}$ for each of the phases $1\le\ell\le k$.
  Then Lemma~\ref{lem:alpha:bound} would imply that at some point $M_{\ell,i}$ is sufficiently large.
	Thus we have $\max_{i \in [r]}\alpha_{\ell,i}\ge\alpha_0$.

We have $A_{\ell}=\max_{i \in [r]}b_{\ell}\beta_{\ell,i}|M_{\ell,i}|$, $b_{\ell}=\frac{1}{2\ell+2}$, and $\beta_{\ell,i}|M_{\ell,i}|\sim \mathcal{B}\left(\alpha_{\ell,i}|M_{\ell,i}|,\frac{1}{2{(2\ell)}^\ell}\right)$.
Now for $\ell\le k$, we have that $\frac{1}{2{(2\ell)}^\ell}\ge\frac{1}{2{(2k)}^k}$.
Thus $\max_{i \in [r]}\alpha_{\ell,i}\ge\alpha_0$ and $|M_{\ell,i}|\ge|M_{\ell,1}|$ implies that the distribution of $A_{\ell}$ statistically dominates the distribution of $b_{\ell}\cdot\mathcal{B}\left(\alpha_{0}|M_{\ell,1}|,\frac{1}{2{(2k)}^k}\right)$.
  Hence if we define $X_{\ell}$ to be independent random variables distributed as $\mathcal{B}\left(\alpha_{0}|M_{\ell,1}|,\frac{1}{2{(2k)}^k}\right)$ for each $\ell\in[r]$, then the distribution of $A_{\ell}$ statistically dominates the distribution of $b_{\ell}\cdot X_{\ell}$.
	Thus,
  \begin{align*}
  \Pr\left[\sum_{\ell\in[k]}\max_{i\in[r]}b_{\ell}\beta_{\ell,i}\ge2+r\delta\right]&\ge\Pr\left[\sum_{\ell\in[k]}\max_{i\in[r]}b_{\ell}\beta_{\ell,i}|M_{\ell,i}|\ge(2+r\delta)|M_{k,r}|\right]\\
	&=\Pr\left[\sum_{\ell\in[k]}A_{\ell}\ge(2+r\delta)|M_{k,r}|\right]\\
		&\ge\Pr\left[\sum_{\ell\in[k]}X_{\ell}\ge\frac{4+2r\delta}{b_k}\cdot|M_{1,1}|\right],
  \end{align*}
	where the final inequality results from $M_{1,1}$ being a maximal matching and $b_{\ell}\ge b_k$ for $\ell\in[k]$.

	Now the variables $X_{\ell}$ are independent but not identically distributed.
	Nevertheless, we can write $Y=\sum_{\ell\in[k]}X_{\ell}$ and note that the distribution of $Y$ statistically dominates the distribution of $Z\sim\mathcal{B}\left(\alpha_0|M_{1,1}|r,\frac{1}{2{(2k)}^k}\right)$ since $|M_{\ell,i}|\ge|M_{1,1}|$ for all $\ell\in[k]$ and $i\in[r]$.
  Thus for $b_k=\frac{1}{2k+2}$ and $\delta=\frac{b_k}{r}$,
  \begin{align*}
  \Pr\left[\sum_{\ell\in[k]}X_{\ell}\ge\frac{4+2r\delta}{b_k}\cdot|M_{1,1}|\right]&\ge\Pr\left[Z\ge\frac{4+2b_k}{b_k}|M_{1,1}|\right]\\
	&=\Pr[Z\ge|M_{1,1}|(8k+10)].
  \end{align*}
  For $r=\frac{2{(2k)}^k(16k+20)}{\alpha_0}$, we have that $\E[Z]=(16k+20)|M_{1,1}|$.
  Thus from a simple Chernoff bound,
  \begin{align*}
  \Pr[Z\ge|M_{1,1}|(8k+10)]=1-\Pr\left[Z<\frac{\E[Z]}{2}\right]>1-e^{-2(16k+20)|M_{1,1}|}\ge 0.99.
  \end{align*}
  Putting things together, we have
  \begin{align*}
  \Pr\left[|M|\ge2|M_{1,1}|\right]\ge 0.99,
  \end{align*}
  which implies that there exists a maximum matching with more than double the number of edges of a maximal matching and is a contradiction.
  Therefore, our assumption that $\max_{i \in [r]}\alpha_{\ell,i}\ge\alpha_0:=\frac{1}{2k^2(k-1)}$ for each of the phases $1\le\ell\le k$ must have been invalid.
  However, if $\alpha_{\ell,i}<\frac{1}{2k^2(k-1)}$ for some $i\in[r]$ and $\ell\in[k]$, then by Lemma~\ref{lem:alpha:bound}, we have for $k=\lceil\epsilon^{-1}+1\rceil$ and sufficiently small $\epsilon>0$ that
  \[
    |M_{\ell,i}|\ge\frac{|M^*|}{1+1/k}\ge\frac{|M^*|}{1+\epsilon}\ge(1-\epsilon)|M^*|,
  \]
  with probability at least $0.99$.
	Thus, Algorithm~\ref{alg:randomized-matching} outputs a $(1-\epsilon)$ approximation the maximum matching with probability at least $0.99$.
\end{proof}

\paragraph{Boosting the Success Probability.}
To increase the probability of success to $1-p$ for any $p\in(0,1)$, a na\"{\i}ve approach would be to run $O\left(\log\frac{1}{p}\right)$ iterations of Algorithm~\ref{alg:randomized-matching} in parallel.
However, the sensitivity analysis becomes considerably more challenging.
Instead, we note that the $\frac{1}{e}$ probability of failure is actually a significant weakening of the $e^{-2(16k+20)|M_{1,1}|}$ probability of failure.
Thus, increasing $k$ by a factor of $O\left(\log\frac{1}{p}\right)$ increases the probability of success to $1-p$.
However, for subconstant $p$, it also substantially increases the asymptotic sensitivity of Algorithm~\ref{alg:randomized-matching}.

\subsection{Sensitivity of the Randomized Greedy and Algorithm~\ref{alg:randomized-matching}}\label{subsec:randomized-sensitivity}
To analyze the sensitivity of Algorithm~\ref{alg:randomized-matching}, we first analyze the sensitivity of the randomized greedy algorithm.

\subsubsection{Sensitivity of the Randomized Greedy}\label{sec:greedy}
In this section, we study the sensitivity of the randomized greedy with respect to vertex deletions.
Recall that given a graph $G=(V,E)$, the randomized greedy works as follows.
First, it chooses a random ordering $\pi$ over edges.
Then starting with an empty matching $M$, it iteratively adds the $i$-th edge in the ordering $\pi$ to $M$ if it is not adjacent to any edge in $M$.
The main result of this section is the following.
\begin{theorem}\label{thm:randomized-greedy}
  Let $A$ be the randomized greedy for the maximum matching problem.
  Then, for any graph $G=(V,E)$ and a vertex $v \in V$, we have
  \[
    d_{\mathrm{EM}}(A(G),A(G-v)) \leq 1.
  \]
\end{theorem}
We need to consider deleting vertices to analyze the sensitivity of our randomized $(1-\epsilon)$-approximation algorithm for the maximum matching problem in Section~\ref{sec:randomized}.

Our analysis is a slight modification of a similar result for the maximal independent set problem~\cite{censor2016optimal}. 
Hence, we defer the proof to Appendix~\ref{apx:randomized-greedy}.

\subsubsection{Sensitivity of Algorithm~\ref{alg:randomized-matching}}
We first analyze the sensitivity of \Call{AugmentingPaths}{}.
Let us fix the graphs $G=(V,E)$ and $G'=(V,E')$, and matchings $M \subseteq E$ and $M' \subseteq E'$, a positive integer $\ell$, and $\delta > 0$.
Let $H$ and $H'$ be the layered graphs constructed using $G,M,\ell$ and $G',M',\ell$, respectively, and let $\tilde{L}$ and $\tilde{L}'$ be the set of active vertices in $H$ and $H'$, respectively.

\begin{lemma}\label{lem:sensitivity-active-vertex-set}
  We have $d_{\mathrm{H}}(\tilde{L}, \tilde{L}') \leq 3d_{\mathrm{H}}(E, E') + 3d_{\mathrm{H}}(M,M')$.
\end{lemma}
\begin{proof}
  Each edge modification in the graph or the matching may cause activate/inactivate at most three vertices in the layered graph (two of them are in the first and last layers, and the remaining one is in one of the middle layers), and hence the lemma follows.
\end{proof}

For two tag functions $t,t'\colon V(H) \to V(H) \cup \{\mathsf{untagged}, \mathsf{dead end}\}$, we define $d_{\mathrm{H}}(t,t') = |\{v \in V(H) \mid t(v) \neq t'(v)\}|$.
We will use symbols $t$ and $t'$ to denote tag functions for $H$ and $H'$, respectively.
Note that the supposed domain of $t'$ is $V(H')$, but it is equal to $V(H)$.
\begin{lemma}\label{lem:augmenting-paths-sensitivity}
  Let $\mathcal{P} = \Call{AugmentingPaths}{G,M,\ell,\delta}$ and $\mathcal{P}' = \Call{AugmentingPaths}{G',M',\ell,\delta}$.
  Then, we have
  \[
    d_{\mathrm{EM}}\left(\mathcal{P}, \mathcal{P}'\right) \leq (d_{\mathrm{H}}(E, E') + d_{\mathrm{H}}(M, M')) \cdot 3^K,
  \]
  where $K = {(1/\epsilon)}^{2^{O(1/\epsilon)}}$.
\end{lemma}
\begin{proof}
  Let $A$ and $A'$ denote $\Call{AugmentingPaths}{G,M,\delta,\ell}$ and $\Call{AugmentingPaths}{G',M',\delta,\ell}$, respectively.
  Let $M_1,M_2,\ldots,M_K$ and $M'_1,M'_2,\ldots,M'_K$ be the sequences of matchings constructed during the process of $A$ and $A'$, respectively.
  Note that $A$ and $A'$ construct the same number of matchings, and that $K \leq {{(1/\epsilon)}^{2^{O(1/\epsilon)}}}$, which follows by a similar argument to that in the proof of Lemma~\ref{lem:time-complexity}.
  For $i \in [K]$, let $S_i$ and $S'_i$ be the vertex sets on which $M_i$ and $M'_i$, respectively, are constructed, that is, the vertex set passed on to $\Call{RandomizedGreedy}{}$, and let $t_i$ and $t'_i$ be the tag functions right before constructing $M_i$ and $M'_i$, respectively.
  By Lemma~\ref{lem:sensitivity-active-vertex-set}, we have $d_{\mathrm{EM}}(S_1, S'_1) \leq c$, where $c = 3d_{\mathrm{H}}(M, M') + 3d_{\mathrm{H}}(E, E')$.
  First, because each difference between $M_{i-1}$ and $M'_{i-1}$ increases the Hamming distance between $t_i$ and $t'_i$ by one, we have
  \begin{align*}
    d_{\mathrm{EM}}(t_i, t'_i) &\leq d_{\mathrm{EM}}(M_{i-1}, M'_{i-1}) + d_{\mathrm{EM}}(t_{i-1}, t'_{i-1}) \leq \cdots \\
		&\leq \sum_{j=1}^{i-1}d_{\mathrm{EM}}(M_{j}, M'_{j}) + d_{\mathrm{EM}}(t_1, t'_1) = \sum_{j=1}^{i-1}d_{\mathrm{EM}}(M_{j}, M'_{j}).
  \end{align*}
  Then we have
  \[
    d_{\mathrm{EM}}(S_{i}, S'_{i}) \leq d_{\mathrm{EM}}(M_{i-1}, M'_{i-1}) + d_{\mathrm{EM}}(t_i, t'_i)
    \leq
    2\sum_{j=1}^{i-1} d_{\mathrm{EM}}(M_{i}, M'_{i}) \leq
    2 \sum_{j=1}^{i-1}d_{\mathrm{EM}}(S_j, S'_j),
  \]
  where the last inequality is due to Theorem~\ref{thm:randomized-greedy}.
  Solving this recursion, we get
  \[
    \sum_{j=1}^i d_{\mathrm{EM}}(S_j, S'_j) \leq c \cdot 3^{i-1},
  \]
  and hence we have $d_{\mathrm{EM}}(S_K, S'_K) = 2 \cdot 3^{K-2}  c$, and the claim follows.
\end{proof}

\noindent
We now show that the sensitivity of Algorithm~\ref{alg:randomized-matching} is $O_{\epsilon}(1)$.
\begin{theorem}\label{thm:randomized:sensitivity}
  The sensitivity of Algorithm~\ref{alg:randomized-matching} is at most $3^K$, where $K = {(1/\epsilon)}^{2^{O(1/\epsilon)}}$.
\end{theorem}
\begin{proof}
  Let $G = (V,E)$ be a graph and $G' = (V,E') = G - e$ for some $e \in E$.
  Let $M_0,M_1,\ldots,M_{kr}$ be the sequence of matchings we construct in Algorithm~\ref{alg:randomized-matching} on $G$, where $M_0$ is the matching constructed at Line~\ref{line:randomized-matching-0}, and $M_j$ is the matching constructed at Line~\ref{line:randomized-matching-0} in the round $i$ of the phase $\ell$ such that $j = (\ell-1) r + i$.
  We define $M'_0,M'_1,\ldots,M'_{kr}$ similarly using $G'$.
  Then, we have by Theorem~\ref{thm:randomized-greedy}
  \[
    d_{\mathrm{EM}}(M_0,M'_0) \leq 1,
  \] and we have by Lemma~\ref{lem:augmenting-paths-sensitivity}
  \begin{align*}
    d_{\mathrm{EM}}(M_i,M'_i)
    & \leq d_{\mathrm{EM}}(M_{i-1},M'_{i-1}) + (d_{\mathrm{EM}}(E,E') + d_{\mathrm{EM}}(M_{i-1},M'_{i-1})) \cdot 3^K \\
    & = d_{\mathrm{EM}}(M_{i-1},M'_{i-1}) (3^K + 1) + 3^K
  \end{align*}
  for $i \in [kr]$, where $K = {(1/\epsilon)}^{2^{O(1/\epsilon)}}$.
  Solving the recursion, we get
  \[
    d_{\mathrm{EM}}(M_{kr},M'_{kr}) \leq 2 {\left(1 + 3^K\right)}^{kr} - 1,
  \]
  and we have the desired bound.
\end{proof}

\noindent
The proof of Theorem~\ref{thm:matching} then follows from Theorem~\ref{thm:randomized-approximation} and Theorem~\ref{thm:randomized:sensitivity}.

\paragraph{Sensitivity to Vertex Deletions.}
We remark that Algorithm~\ref{alg:randomized-matching} also has sensitivity $O(3^K)$, for $K = {(1/\epsilon)}^{2^{O(1/\epsilon)}}$, to vertex deletions.
Recall that Lemma~\ref{lem:sensitivity-active-vertex-set} crucially relies on each edge deletion changing at most three vertices in the layered graph.
That is, due to the construction of the layered graph, each edge deletion changes at most two altered vertices in the first and last layers, and at most one altered vertex in one of the middle layers.
This is because the first layer and the last layer encode the vertex set $V$, while each matched edge is assigned to one of the middle layers.

Observe that when we delete a vertex $v$, at most one vertex in the vertex set $V$ is altered, so that the first and last layer of the layered graph each have one change.
Moreover, at most one matched edge is incident to $v$, so at most one vertex in one of the middle layers is altered as well.
Thus, at most three vertices in the layered graph are changed as a result of the vertex deletion, so the sensitivity of Algorithm~\ref{alg:randomized-matching} to vertex deletions is again $O(3^K)$, for $K = {(1/\epsilon)}^{2^{O(1/\epsilon)}}$.

\subsection{Applications to Online Matching with Replacements}\label{subsec:online-matching}
In this section, we show that Algorithm~\ref{alg:randomized-matching} can be repurposed to obtain an algorithm for the online matching problem with replacements.
In the \emph{edge-arrival} model for the online matching problem with replacements, the edges $E$ of the graph $G=(V,E)$ arrive sequentially as a data stream, and the goal is to maintain or approximate a maximum matching across all times, while minimizing the number of total edges that are altered between successive outputs of the algorithm.
Formally, let $E_i=(e_1,\ldots,e_i)$ be the subset of edges of the graph that have arrived by time $i$ and let $M^*_i$ be the maximum matching on $G_i=(V,E_i)$.
Given a constant $c\le 1$, the goal of the online matching problem with replacements is to output a sequence of matchings $M_1,\ldots,M_{|E|}$ that minimizes $\sum_{i=1}^{|E|-1}d_{\mathrm{H}}(M_{i+1}, M_i)$ subject to the constraint $|M_i|\ge c|M^*_i|$, i.e., each matching $M_i$ is a $c$-approximation to the maximum matching at time $i$.
The quantity $d_{\mathrm{H}}(M_{i+1}, M_i)$ is the number of replacements at time $i$ and the quantity $\sum_{i=1}^{|E|-1}d_{\mathrm{H}}(M_{i+1}, M_i)$ is the total number of replacements.
The \emph{vertex-arrival} model is defined analogously, with the exception that the stream updates are a vertex $v_i$, along with all the edges adjacent to $v_i$.

Bernstein~et~al.\cite{BernsteinHR19} gives an algorithm for online \emph{bipartite} matching with replacements in the vertex-arrival model that always outputs a maximum matching but has $O(n\log^2 n)$ total replacements.
We show that our algorithm can be modified to achieve total replacements $O_{\epsilon}(n)$ and $(1-\epsilon)$-approximate maximum matchings for general graphs, i.e., not just bipartite graphs.

The challenge to immediately applying Algorithm~\ref{alg:randomized-matching} to the online matching with replacements setting is that the guarantee of Theorem~\ref{thm:matching} is only in terms of earth-mover's distance.
Thus, we cannot apply a black-box reduction to the online model because each time we call Algorithm~\ref{alg:randomized-matching}, we can obtain a completely different matching, depending on the randomness of the algorithm.
For example, suppose Algorithm~\ref{alg:randomized-matching} guarantees that each time $i$ of the stream, there exist two maximal matchings $M_{i,1}$ and $M_{i,2}$ that are each $(1-\epsilon)$-approximations to the maximum matching, but $d_{\mathrm{H}}(M_{i,1}, M_{i,2})=\Omega(n)$.
Moreover, suppose that at each time, Algorithm~\ref{alg:randomized-matching} outputs $M_{i,1}$ with probability $\frac{1}{2}$ and $M_{i,2}$ with probability $\frac{1}{2}$.
If $d_{\mathrm{H}}(M_{i,1}, M_{i+1,1})=d_{\mathrm{H}}(M_{i,2}, M_{i+1,2})=O(1)$ at all times, then Algorithm~\ref{alg:randomized-matching} has $O(1)$ sensitivity at all times in the stream, but if $\tilde{M_i}$ is the matching output by the algorithm at each time $i$, we could potentially have $d_{\mathrm{H}}(\tilde{M_{i}}, \tilde{M_{i+1}})=\Omega(n)$ replacements, so that the total number of replacements is $\Omega(n^2)$.
Instead, we open up the black-box of Algorithm~\ref{alg:randomized-matching} and show that we can achieve $O_{\epsilon}(n)$ total replacements by fixing components of the internal randomness of the algorithm across the duration of the stream.

\begin{proof}[Proof of Theorem~\ref{thm:mm:recourse}]
Recall that Algorithm~\ref{alg:randomized-matching} first fixes a random permutation $\pi$ of the edges in the subroutine $\textsc{RandomizedGreedy}$.
Equivalently, we can fix a random permutation $\pi$ of $\binom{n}{2}$, which induces a consistent permutation of the edges across the entire stream.
Let $A_{\pi}$ be the deterministic algorithm obtained from our randomized algorithm after sampling $\pi$ uniformly at random at the beginning of the stream and fixing the permutation $\pi$ of edges afterwards.
Then whenever a new vertex arrives, we simply run $A_{\pi}$ on the current graph and return the solution.
Since the expected number of replacements at each time is $O_{\epsilon}(1)$, then the expected total number of replacements is $O_{\epsilon}(n)$.
Thus by Markov's inequality, the total number of replacements will be $O_{\epsilon}(n)$ with probability 0.99.
\end{proof}

\section{Deterministic Maximal Matching for Bounded-Degree Graphs}\label{sec:deterministic}
In this section, we give a deterministic algorithm for computing a maximal matching that has low sensitivity on bounded-degree graphs.
The main idea is to use deterministic local computation algorithms with a small number of probes to find a maximal matching.
Our algorithm uses two main ingredients.
The first ingredient is a deterministic LCA of~\cite{ColeV86} for $6^{\Delta}$-coloring a graph with maximum degree $\Delta$, using $O(\Delta\log^*n)$ probes.
The second ingredient is a framework of~\cite{ParnasR07} that simulates local distributed algorithms using a deterministic LCA\@.
In particular, we use the framework to simulate an algorithm that takes a coloring of a graph and outputs a maximal matching.
We give the details for the local distributed algorithm in Algorithm~\ref{alg:det:mm:color}.

To bound the sensitivity of the algorithm, it suffices to analyze the number of queries for which a deleted edge would be probed.
Crucially, both the deterministic LCA of~\cite{ColeV86} and the framework of~\cite{ParnasR07} only probe edges (incident to vertices) within a small radius of the query.
Thus, only a small number of queries will probe the edge that is altered, so that the output of the algorithm only has a small number of changes.

We first require a deterministic LCA of~\cite{ColeV86} for $6^{\Delta}$-coloring a graph with degree $\Delta$, using a small number of probes within distance $O(\Delta\log^*n)$ of the query.
We give the full details in Algorithm~\ref{alg:coloring}, which has the following guarantee.
\begin{lemma}[\cite{ColeV86}]\label{lem:lca:coloring}
There exists a deterministic LCA \textsc{ColoringLCA} for $6^{\Delta}$-coloring a graph with degree $\Delta$, using $O(\Delta\log^*n)$ probes.
\end{lemma}

\begin{algorithm}[t!]
  \caption{LCA Algorithm \textsc{ColoringLCA} for $6^\Delta$-coloring with $O(\Delta\log^* n)$ probes}\label{alg:coloring}
  \Procedure{\emph{\Call{FormForests}{$G,\Delta$}}}{
	//Decompose graph into $\Delta$ oriented forests\;
	\For{$i=1$ to $\Delta$}{
		Let $N_u(i)$ denote the $i$-th neighbor of $u$ according to the IDs of the vertices\;
		Let $E_i=\{(u,v):\mathsf{id}(u)<\mathsf{id}(v), v=N_u(i)\}$\;
		Let $G_i = (V, E_i)$ be the oriented tree, where the root has no out-going edges\;
	}
	}
  \Procedure{\emph{\Call{ColorForests}{$G_i,\Delta$}}}{
	\For{$\Theta(\log^* n)$ rounds}{
		\For{each nodes $u$}{
			\If{$u$ is a root node}{
				Set $\phi_u$ to $0$\;}
			\Else{
				Let $v$ be the parent of $u$ in $G_i$\;
				Let $a_u$ be the index of the least significant bit with $\phi_u\neq\phi_v$\;
				Let $b_u$ be the value of the $a_u$-th bit of $u$\;
				$\phi_u\gets a_u\circ b_u$.
				}
			}
		}
	}
\end{algorithm}

We now describe a local distributed algorithm that takes a coloring of a graph and outputs a maximal matching.
The algorithm iterates over all colors and adds any edge adjacent to a vertex of a particular color to the greedy matching if there is no other adjacent edge already present in the matching.
We give the algorithm in full in Algorithm~\ref{alg:det:mm:color}.
\begin{algorithm}[t!]
  \caption{Maximal Matching Algorithm \textsc{Coloring-to-MM}}\label{alg:det:mm:color}
  \Procedure{\emph{\Call{Coloring-to-MM}{$G$ colored with $c$ colors}}}{
	$M\gets\emptyset$\;
	\For{color $i=1$ to $i=c$}{
		\If{edge $(u,v)$ has either $\phi_u=i$ or $\phi_v=i$}{
			Add $(u,v)$ to $M$ if no adjacent edge is in $M$.
			}
		}
	}
\end{algorithm}

Putting things together, we obtain a deterministic maximal matching algorithm in Algorithm~\ref{alg:det:mm:}.
\begin{algorithm}[t!]
  \caption{Maximal Matching Algorithm}\label{alg:det:mm:}
  \Procedure{\emph{\Call{Coloring-to-MM}{Graph $G$}}}{
	Coloring $G_C=\textsc{ColoringLCA}(G,\Delta)$\;
	Output $\textsc{Coloring-to-MM}(G_C)$\;
	}
\end{algorithm}

We next require the following framework of~\cite{ParnasR07} that simulates local distributed algorithms using a deterministic LCA.
In particular, we will implement Algorithm~\ref{alg:det:mm:color}.
\begin{lemma}[\cite{EvenMR15,ParnasR07}]\label{lem:lca:colortomm}
Given access to an oracle that takes vertices of an underlying as queries and outputs a color for the queried vertex, there exists a deterministic LCA that can implement \textsc{Coloring-to-MM} using $\Delta^{O(c)}$ probes.
\end{lemma}
Given a \textsc{ColoringLCA} for $6^\Delta$-coloring, the LCA for maximal matching in Lemma~\ref{lem:lca:colortomm} uses the following idea.
For a query edge $e$, we first call \textsc{ColoringLCA} for every vertex with distance roughly $6^\Delta$ from $e$.
Parnas and Ron~\cite{ParnasR07} then shows it suffices to run Algorithm~\ref{alg:det:mm:color} locally on the graph of radius roughly $6^\Delta$ from $e$.

We now show that our deterministic LCA based algorithm outputs a maximal matching with low worst case sensitivity for low-degree graphs.

\begin{proof}[Proof of Theorem~\ref{thm:det:matching}]
	Consider running the deterministic LCA from Lemma~\ref{lem:lca:colortomm} that simulates \textsc{Coloring-to-MM} on a graph $G$ and a graph $G':=G-e$, for some $e \in E$.
	Let $S$ be the set of vertices that are assigned different colors in $G$ and $G'$ by \textsc{ColoringLCA}.
	First observe that $e$ is within distance $\Theta(\log^*n)$ from at most $\Delta^{\Theta(\log^*n)}$ other vertices.
	Hence, from Lemma~\ref{lem:lca:coloring} we have $|S|\le\Delta^{\Theta(\log^*n)}$.
	Moreover, each vertex $u\in S$ is within distance $O\left(6^\Delta\right)$ from at most $\Delta^{O\left(6^\Delta\right)}$ other vertices.
	Thus the total number of edges that differ between the matchings $M$ and $M'$ output by \textsc{Coloring-to-MM} for $G$ and $G'$ respectively is at most
	\[
		\Delta^{\Theta(\log^*n)}\cdot\Delta^{O\left(6^\Delta\right)}=\Delta^{O\left(6^\Delta+\log^* n\right)}.
		\qedhere
	\]
\end{proof}

\noindent
It is clear that an almost identical analysis goes through for vertex sensitivity.


\section{Lower Bounds for Maximum Matching}\label{sec:lb}
In this section, we show lower bounds for deterministic and randomized algorithms for the maximum matching problem. 
\subsection{Deterministic Lower Bound}\label{sec:det:lb}
In this section, we prove Theorem~\ref{thm:lb-deterministic}, which claims that \emph{any} deterministic algorithm for the maximum matching problem has edge sensitivity $\Omega(\log^*n)$.
Our proof relies on Ramsey's theorem. First we introduce some definitions.
Let $Y$ be a finite set.
We say that $X$ is a \emph{$k$-subset} of $Y$ if $X \subseteq Y$ and $|X|=k$.
Let $Y^{(k)} =\{X \subseteq Y \mid |X|=k\}$ be the collection of all $k$-subsets of $Y$.
A $c$-labeling of $Y^{(k)}$ is an arbitrary function $f : Y^{(k)} \to [c]$.
Then we say that $X \subseteq Y$ is \emph{monochromatic} in $f$ if $f(A)= f(B)$ for all $A,B \in X^{(k)}$.
Let $R_c(n; k)$ be the smallest integer $N$ such that the following holds:
for any set $Y$ with at least $N$ elements, and for any $c$-labeling $f$ of $Y^{(k)}$, there is an $n$-subset of $Y$ that is monochromatic in $f$.
If no such $N$ exists, $R_c(n; k) = \infty$.
Define $\mathrm{twr}(k)$ as the tower of twos of height $k$, that is, $\mathrm{twr}(1) = 2$ and $\mathrm{twr}(k+1) = 2^{\mathrm{twr}(k)}$.
We will use the following formulation of Ramsey's theorem.
\begin{theorem}[Special case of Theorem~1 in~\cite{Erdos1952}]\label{thm:ramsey}
  For any positive integer $t$, $R_{2^t}(t+1; t) \leq \mathrm{twr}(O(t))$.
\end{theorem}

We now show that any deterministic constant-factor approximation algorithm for the maximum matching problem has edge sensitivity $\Omega(\log^* n)$.
\begin{proof}[Proof of Theorem~\ref{thm:lb-deterministic}]
  Let $A$ be an arbitrary deterministic algorithm that outputs a maximal matching, let $t$ be a positive integer, which will be determined later.
  Let $\mathcal{G}$ be a class of graphs on the vertex set $[n]$ consisting of a cycle $v_1,\ldots,v_t$ with $v_1 < v_2 < \cdots < v_t$ and $n-t$ isolated vertices.
  Given a matching $M$ on the cycle $v_1,\ldots,v_t$, we encode it to an integer $0 \leq k \leq 2^t-1$ so that the $i$-th bit of $k$ is $1$ if and only if the edge $\{v_i, v_{i+1}\}$ belongs to $M$, where we regard $v_{t+1} = v_1$.
  Then, we can regard the algorithm $A$ as a function $f: \binom{[n]}{t} \to \{0,1,\ldots,2^t-1\}$, that is, given a set $\{v_1,\ldots,v_t\} \subseteq [n]$ with $v_1<v_2<\cdots<v_t$, we compute a matching on the cycle $v_1,\ldots,v_t$, and encode it to an integer.
  Then if $n \geq R_{2^t}(t+1; t)$, which holds when $t = O(\log^*n)$ by Theorem~\ref{thm:ramsey}, there exists a set $S = \{s_0,s_1,\ldots,s_{t}\} \subseteq [n]$ with $s_0 < s_1 < \cdots <s_t$ such that $f(T)$ is constant whenever $T \subseteq S$ with $|T|=t$.
  Let $G,G' \in \mathcal{G}$ be the graph with cycles $s_0,\ldots,s_{t-1}$ and $s_1,\ldots,s_{t}$, respectively, and let $M$ and $M'$ be the matching output by $A$ on $G$ and $G'$, respectively.
  As $M$ and $M'$ have the same encoding, $\{s_i,s_{i+1 \bmod t} \} \in M$ if and only if $\{s_{i+1},s_{i+2}\} \in M'$, where we regard $s_{t+2}=s_1$.
  Note that, however, if $\{s_i,s_{i+1 \bmod t} \} \in M$ then $\{s_{i+1 \bmod t},s_{i+2 \bmod t} \} \not \in M'$.
  It follows that $d_{\mathrm{H}}(M,M') = \Omega(|M|) = \Omega(t)$, where the last equality holds because $A$ has a constant approximation ratio.
\end{proof}

\subsection{Lower Bounds for Deterministic Greedy Algorithm}\label{subsec:lb-deterministic-greedy}
As we have seen in Section~\ref{sec:greedy}, the randomized greedy algorithm has $O(1)$ sensitivity even for vertex deletion.
Can we derandomize it without increasing the sensitivity?
To make the question more precise, let $V$ be a set of $n$ vertices and $\pi$ be a permutation over $\binom{V}{2}$.
Then, let $A_\pi$ denote the greedy algorithm such that, starting with an empty matching $M$, it iteratively adds the $i$-th edge with respect to $\pi$ to $M$ if and only if the edge does not share an endpoint with any edge in $M$.
We now show that the answer to the question is negative.
\begin{theorem}\label{thm:lb-greedy}
For any permutation $\pi$ over $\binom{V}{2}$, the algorithm $A_\pi$ has sensitivity $\Omega(n)$.
\end{theorem}
\begin{proof}
  We say that an element $e$ of a poset \emph{covers} another element $e'$ if $e>e'$, where $>$ is the order relation of the poset, and there is no other element $e''$ such that $e>e''>e'$.
  Then, we construct a poset $P$ on the element set $\binom{V}{2}$ in which a pair $e \in \binom{V}{2}$ covers another pair $e' \in \binom{V}{2}$ if $\pi(e) > \pi(e')$ and $|e \cap e'| \geq 1$.
  Note that the size of any antichain in $P$ is at most $n/2$: A set of elements of size more than $n/2$ must have two elements $e,e'$ with $|e \cap e'| \geq 1$, which form a chain of length two.
  Hence, we need at least $\binom{n}{2} / (n/2) = n-1$ antichains to cover all the elements in $P$.
  Then by Mirsky's theorem, there exists a chain, say, $e_1,\ldots,e_{n-1}$, of size $n-1$ in $P$.
  From the construction of $P$, $e_1,\ldots,e_{n-1}$ forms a path of length $n-1$.
  Then, $A_\pi$ on the path $e_1,\ldots,e_{n-1}$ outputs edges with odd indices, whereas $A_\pi$ on the path $e_2,\ldots,e_{n-1}$ outputs edges with even indices, and hence the sensitivity of $A_\pi$ is $\Omega(n)$.
\end{proof}

\subsection{Lower Bounds for Randomized Algorithms}\label{subsec:lb-randomized}
The following shows that sensitivity must increase as approximation ratio goes to one.
\begin{theorem}\label{thm:lb-randomized}
 Let $\epsilon > 0$.
 Any (possibly randomized) $(1-\epsilon)$-approximation algorithm for the maximum matching problem has sensitivity $\Omega(1/\epsilon)$.
\end{theorem}
\begin{proof}
  For simplicity, we assume $1/10\epsilon$ is an even integer.
  Let $A$ be an arbitrary $(1-\epsilon)$-approximation algorithm for the maximum matching problem, and let $G$ be a graph consisting of a cycle of length $1/10\epsilon$ and $n-1/10\epsilon$ isolated vertices.
  Clearly $G$ has two disjoint maximum matchings, say, $M_1,M_2$, of size $1/20\epsilon$.
  Let $p_1$ and $p_2$ be the probability that $A$ on $G$ outputs $M_1$ and $M_2$, respectively.
  Then as $A$ has approximation ratio $1-\epsilon$, we have
  \[
    p_1 \cdot \frac{1}{20\epsilon} +
    p_2 \cdot \frac{1}{20\epsilon} +
    (1-p_1-p_2) \cdot \left(\frac{1}{20\epsilon}-1\right)
    \geq  \frac{1-\epsilon}{20\epsilon}.
  \]
  Hence, we have $p_1+p_2 \geq 19/20$, and it follows that at least one of $p_1 \geq 19/40$ and $p_2 \geq 19/40$ hold.
  Without loss of generality, we assume $p_1 \geq 19/40$.

  Let $G'$ be the graph obtained from $G$ by removing one edge in $M_1$.
  Then, $G'$ has a unique maximum matching $M_2$.
  Let $p'_2$ be the probability that $A$ on $G'$ outputs $M_2$.
  As $A$ has approximation ratio $1-\epsilon$, we have
  \[
    p_2 \cdot \frac{1}{20\epsilon} + (1-p_2) \cdot \left( \frac{1}{20\epsilon}-1 \right) \geq \frac{1-\epsilon}{20\epsilon},
  \]
  which implies
  $p'_2 \geq 19/20$.
  Hence, the sensitivity of $A$ is at least
  \[
    \max\Bigl(\Pr[A(G) = M_1] - \Pr[A(G') \neq M_2],0\Bigr) \cdot d_{\mathrm{H}}(M_1,M_2)
    \geq \left(\frac{19}{40} - \frac{1}{20}\right) \cdot \frac{1}{10\epsilon}
    = \Omega\left(\frac{1}{\epsilon}\right).
    \qedhere
  \]
\end{proof}


\section{Weighted Sensitivity and Maximum Weighted Matching}
\label{sec:weighted}
In this section, we consider a generalization of sensitivity to weighted graphs, and show an approximation algorithm with low sensitivity for the maximum weighted matching problem.

\subsection{Weighted Sensitivity}
Given a weight function over the edges $w:E\to\mathbb{R}$ of a graph $G=(V,E)$ and two edge sets $S$ and $S'$, let
\[d_{\mathrm{H}}^w(S,S') = \sum_{e\in S \triangle S'}w(E),\]
where $\triangle$ again denotes the symmetric difference.
For random edge sets $X$ and $X'$, we use $d_{\mathrm{EM}}^w(X,X')$ to denote the weighted earth mover's distance between $X$ and $X$ with respect to $w$, so that
\[d_{\mathrm{EM}}^w(X,X') = \min_{\mathcal{D}} \E_{(S,S') \sim \mathcal{D}}d_{\mathrm{H}}^w(S,S'),\]
be the weighted Hamming distance between $S$ and $S'$ with respect to $w$, where $\mathcal{D}$ is a distribution such that its marginal distributions on the first and second coordinates are $X$ and $X'$, respectively.
For a real-valued function $\beta$ on graphs, we say that the \emph{weighted sensitivity} of an algorithm $A$ that outputs a set of edges is at most $\beta$ if for every graph $G=(V,E)$, a weight function $w:E \to \mathbb{R}$, and an edge $e\in E$,
\[d_{\mathrm{EM}}^w(A(G),A(G-e)) \leq \beta(G).\]
A priori, it is not clear whether the weighted sensitivity of an algorithm should correlate with the weight of removed edges.
Thus we say that the \emph{normalized weighted sensitivity} is at most $\beta$ if
\[ \frac{d_{\mathrm{EM}}^w(A(G),A(G-e))}{w(e)} \leq \beta(G).\]

\subsection{Algorithm Description}
We use a simple approach of partitioning the input by weight, finding a maximal matching on each partition, and finally forming a weighted matching by greedily adding edges from the maximal matchings, beginning with the matchings in the largest weight classes.
The approach is known to give a $(4+\epsilon)$-approximation~\cite{BuryGMMSVZ19,CrouchS14} to the maximum weighted matching.
However to bound the weighted sensitivity of our algorithm, we must choose $\epsilon$ carefully.

Formal description of our algorithm is given in Algorithm~\ref{alg:weighted-matching}.
It first defines subsets of edges $E_i$, where we assume the weight of each edge is polynomially bounded in $n$, so that $1\le w(e)\le n^c$ for some constant $c$.
For a parameter $\alpha>1$, we define $E_i$ to be the subset of edges in $E$ with weight at least $\alpha^i$.
Algorithm~\ref{alg:weighted-matching} first draws a random permutation $\pi$ of the edges and greedily forms a maximal matching $M_i$ on each set $E_i$ induced by $\pi$.
It then greedily adds edges to a maximal matching, starting from the matching of the heaviest weight class and moving downward.
That is, we initialize $M$ to be the empty set and greedily add edges of $M_i$ to $M$, starting with $i=O(\log_{\alpha} n^c)$ and decrementing $i$ after each iteration.

\begin{algorithm}[t!]
  \caption{Algorithm for maximum weighted matching}\label{alg:weighted-matching}
  \Procedure{\emph{\Call{Matching}{$G,\epsilon,w$}}}{
		Let $C$ be a sufficiently large constant and $\alpha>1$ be a trade-off parameter.\;
		Let $\pi$ be a random ordering of the edges of $G$\;
		For each $i=0$ to $i=C\log n$, let $E_i$ be the set of edges with weight at least $\alpha^i$\;
    Let $M_i$ be the greedy maximal matching of $E_i$ induced by $\pi$\;
		$M\gets\emptyset$\;
    \For{$i=C\log n$ to $0$}{
			\For{$e\in M_i$}{
				\If{$e$ is not adjacent to $M$}{
					$M\gets M\cup\{e\}$\;
					}
				}
			}
    \Return $M$.
  }
\end{algorithm}

\begin{theorem}[\cite{BuryGMMSVZ19,CrouchS14}]\label{thm:weighted:approx}
Algorithm~\ref{alg:weighted-matching} gives an $\frac{1}{4\alpha}$-approximation to the maximum weighted matching and uses runtime $O\left(m\log_{\alpha} n\right)$ on a graph with $m$ edges and $n$ vertices.
\end{theorem}

\subsection{Sensitivity Analysis}
We first require the following key structural lemma that we use to prove Theorem~\ref{thm:randomized-greedy} in Appendix~\ref{apx:randomized-greedy}.
\begin{lemma}
\label{lem:deleted:change}
In expectation, the deletion of an edge $e$ alters at most one edge in $M_i$, i.e., at most one edge in $M_i$ is inserted or deleted in expectation.  
(Informal, see Lemma~\ref{lem:expected:change}.)
\end{lemma}
The sensitivity analysis follows from the observation that the deletion of an edge $e$ can only affect the matchings $E_i$ for which $\alpha^i\le w(e)$.
Moreover by Lemma~\ref{lem:deleted:change}, the deletion of edge $e$ affects at most one edge in $E_i$, in expectation. 
Thus in expectation, the deletion of $e$ affects at most two edges in $E_{i-1}$ in expectation and inductively, the deletion of $e$ affects at most $2^j$ edges in $E_{i-j}$ in expectation.
On the other hand, the weight of each edge in $E_{i-j}$ is at most $\alpha^{i-j}$ so the weighted sensitivity is $\sum_j \alpha^{i-j}2^j$.
Hence for $\alpha=2$, the weighted sensitivity is $O(2^i\log n)$ and for $\alpha>2$, the weighted sensitivity is $O(2^i)$.
Similarly for normalized weighted sensitivity, we rescale by $\frac{1}{2^i}$ so that the normalized weighted sensitivity is $O(\log n)$ for $\alpha=2$ and $O(1)$ for $\alpha>2$.
We now formalize this intuition. 

\begin{theorem}\label{thm:weighted:sensitivity}
Suppose $1\le w(e)\le W\le n^c$ for some absolute constants $W,c>0$ for all $e\in E$.
The weighted sensitivity of Algorithm~\ref{alg:weighted-matching} is $O(W\log n)$ for $\alpha=2$ and $O(W)$ for $\alpha>2$. 
The normalized weighted sensitivity of Algorithm~\ref{alg:weighted-matching} is $O(\log n)$ for $\alpha=2$ and $O(1)$ for $\alpha>2$.
\end{theorem}
\begin{proof}
Let $e$ be an edge of weight $w(e)\in[2^i,2^{i+1}]$ for some integer $i\ge 0$ and suppose $e$ is removed from $G$.
For $j\le i$, let $S_j$ be the set of edges in $E_j$ affected by the deletion of edge $e$, so that by Lemma~\ref{lem:deleted:change}, $\mathbb{E}[|S_i|]\le 1$.
Then we have $\mathbb{E}[|S_i|\cup\{e\}]\le 2$ so that Lemma~\ref{lem:deleted:change} implies that $\mathbb{E}[|S_{i-1}|]\le 2$.
Now suppose that for a fixed $j\le i$, we have $\mathbb{E}\left[\left|\{e\}\cup\bigcup_{k=j}^{i}S_{k}\right|\right]\le 2^{i-j}$.
Then Lemma~\ref{lem:deleted:change} implies that $\mathbb{E}[|S_{j-1}|]\le 2^{i-j}$ so that
\[\mathbb{E}\left[\left|\{e\}\cup\bigcup_{k=j-1}^{i}S_{k}\right|\right]\le 2^{i-j+1}.\]
Hence by induction, we have $\mathbb{E}[|S_{j}|]\le 2^{i-j}$.

Since each edge of $E_i$ has weight at most $\alpha^i$, then we have
\[d_{\mathrm{EM}}(A(G),A(G - e))=\sum_{j=0}^i 2^{i-j}\alpha^j,\]
where $A(G)$ represents the output of Algorithm~\ref{alg:weighted-matching} on $G$.
Under the assumption that $1\le w(e)\le n^c$ for some absolute constant $c>0$ for all $e\in E$, then $i=O(\log n)$. 
Hence for $\alpha=2$, we have $\sum_{j=0}^i 2^{i-j}\alpha^j=O(2^i\log n)=O(W\log n)$ and for $\alpha>2$, we have $\sum_{j=0}^i 2^{i-j}\alpha^j=O(2^i)=O(W)$.
Moreover, we have
\[\overline{d_{\mathrm{EM}}}(A(G),A(G - e))\le\frac{1}{2^{i+1}}\sum_{j=0}^i 2^{i-j}\alpha^j,\]
so that $\overline{d_{\mathrm{EM}}}(A(G),A(G - e))\le O(\log n)$ for $\alpha=2$ and $\overline{d_{\mathrm{EM}}}(A(G),A(G - e))=O(1)$ for $\alpha>2$.
\end{proof}

\noindent
Together, Theorem~\ref{thm:weighted:approx} and Theorem~\ref{thm:weighted:sensitivity} give the full guarantees of Algorithm~\ref{alg:weighted-matching}.
\begin{theorem}
Let $G=(V,E)$ be a weighted graph with $w(e)\le W\le n^c$ for some constant $c>0$ and all $e\in E$.
For a trade-off parameter $\alpha$, there exists an algorithm that outputs a $\frac{1}{4\alpha}$-approximation to the maximum weighted matching in $O(m\log_{\alpha} n)$ time.
For $\alpha=2$, the algorithm has weighted sensitivity $O(W\log n)$ and normalized weighted sensitivity $O(\log n)$.
For $\alpha>2$, the algorithm has weighted sensitivity $O(W)$ and normalized weighted sensitivity $O(1)$.
\end{theorem}

\noindent
We again emphasize that the worst case weighted sensitivity of \emph{any} constant factor approximation algorithm to the maximum weighted matching problem is at least $\Omega(W)$. 
Recall that if an edge of weight $W=n^c$ is altered in a graph whose remaining edges have weight $1$, then any constant factor approximation to the maximum weighted matching must include the heavy edge for sufficiently large $c>0$, which incurs cost $\Omega(W)$ in the weighted sensitivity. 
Thus for $\alpha>2$, Algorithm~\ref{alg:weighted-matching} performs well with respect to both weighted sensitivity and normalized weighted sensitivity. 
\section{Conclusion and Open Questions}
In this paper, we study the worst-case sensitivity for approximation algorithms for the maximum matching problem. 
We give a randomized $(1-\epsilon)$-approximation algorithm with worst-case sensitivity $O_{\epsilon}(1)$, which improves an algorithm of Varma and Yoshida that offers the same approximation guarantee, but only \emph{average} sensitivity $n^{O(1/(1+\epsilon^2))}$. 
We also give a deterministic $1/2$-approximation algorithm with sensitivity $\exp(O(\log^*n))$ for bounded-degree graphs. 
We introduced the concept of normalized weighted sensitivity for the maximum weighted matching problem and gave an algorithm with $O(1)$ normalized weighted sensitivity that outputs a $\frac{1}{4\alpha}$-approximation to the maximum weighted matching in $O(m\log_{\alpha} n)$ time, for a trade-off parameter $\alpha>2$. 

We believe there are many interesting open questions for future exploration. 
Since our work focuses on the maximum matching problem, we have not considered normalized weighted sensitivity for other graph problems. 
Even for maximum matching, there remains a large number of potential directions for future research. 
For example, there remains a large gap in the understanding of the behavior of the worst-case sensitivity of deterministic algorithms. 
Another line of study is constant factor approximation algorithms for maximum weighted matching with low sensitivity, rather than low normalized weighted sensitivity. 
Can we achieve $(1-\epsilon)$-approximation to the maximum weighted matching problem while still having low normalized weighted sensitivity?

\bibliographystyle{plainurl}
\bibliography{main}

\appendix

\section{Proof of Theorem~\ref{thm:randomized-greedy}}\label{apx:randomized-greedy}
In this section, we formalize the proof of Theorem~\ref{thm:randomized-greedy}.
The approach follows exactly the same structure as the~\cite{censor2016optimal}, who give an algorithm for maximal independent set in the dynamic distributed model.
The only difference is that we maintain a maximal matching rather than a maximal independent set, so we must track the order of the edges rather than the order of the vertices in a given permutation.
We offer the proof for completeness.

In what follows, we fix a graph $G=(V,E)$, $v \in V$, and let $G'=G-v$.
For an permutation $\pi$ over edges in $G$, let $A_\pi$ be the deterministic algorithm that, starting with an empty matching $M$, iteratively add edges to $M$ in the order $\pi$ if they do not intersect with $M$.
Then, the randomized greedy $A$ can be seen as an algorithm that chooses a random permutation $\pi$ and then runs $A_\pi$.
Let $M_\pi$ and $M'_\pi$ be the maximal matchings obtained by running $A_\pi$ on $G$ and $G'$, respectively. (Here we used $\pi$ as an permutation over edges in $G'$ by ignoring $e$ in $\pi$.)
Let $M$ and $M'$ be the maximal matchings obtained by running $A$ on $G$ and $G'$, respectively.

For an edge $e \in E$, we call $\pi(e)$ the \emph{rank} of $e$, and let $I_\pi(e)$ be the set of edges sharing endpoints with $e$ with smaller rank, that is, $I_\pi(e) = \{e' \in N_G(e) \mid \pi(e') < \pi(e)\}$.
Note that the matching $M_\pi$ can be described by the following invariant:
\begin{quote}
  An edge $e$ is in $M_\pi$ if and only if all of its neighbors $e' \in N_G(e) \cap I_\pi(e)$ are not in $M_\pi$.
\end{quote}

Our goal is to show that, in expectation over $\pi$, we need to modify at most one edge in $M_\pi$ so that the invariant is satisfied for the graph $G'$.

Let $e_\pi \in E$ be the edge incident to $v$ with the smallest rank with respect to $\pi$.
We define $S_\pi \subseteq E$ to intuitively be the set of edges in $G$ that need to be changed to maintain the invariant.
Formally, we set $S_{\pi,0} = \{e_\pi\}$ if $e_\pi \in M$ and $S_{\pi,0} = \emptyset$ otherwise.
Then for $i>0$, recursively set
\[
  S_{\pi,i} = \{e \in M_\pi \mid S_{\pi, i-1}\cap I_\pi(e)\neq\emptyset\}\cup \left\{e \not \in M_\pi \mid  I_\pi(e)\cap M_\pi \subseteq\bigcup_{j=0}^{i-1} S_{\pi, j}\right\}.
\]
We then define $S_\pi = \bigcup S_{\pi, i}$ and show the following, from which Theorem~\ref{thm:randomized-greedy} immediately follows.
\begin{lemma}\label{lem:expected:change}
\[\E_\pi|S_\pi| \leq 1.\]
\end{lemma}

We define $S'_\pi \subseteq E$ to intuitively be the set of edges that must be changed to maintain the invariant if $e_\pi$ is moved to the beginning of $\pi$.
That is, $S'_{\pi,0}=\{e_\pi\}$ and we define $S'_{\pi,i}$ using the same recursion as $S_{\pi,i}$, though the underlying permutation is now $\pi$ with $e_\pi$ moved to the beginning.
We then define $S'_\pi=\bigcup S'_{\pi,i}$.
The following is a counterpart of Lemma~2 in~\cite{censor2016optimal}.
\begin{lemma}\label{lem:possible:changes}
  If $\pi(e_\pi ) \neq \min\{\pi(e) \mid e \in S'_\pi\}$, then $S_\pi = \emptyset$.
  Otherwise, $S_\pi \subseteq S'_\pi$.
\end{lemma}
\begin{proof}
  First, suppose that $\pi(e_\pi )\neq \min \{\pi(e) \mid e \in S'_\pi\}$.
  We show that the invariant still holds after the vertex deletion, and thus $S_\pi = \emptyset$.
  Consider the edge $e_{\min} \in E$, for which $\pi(e_{\min}) = \min \{\pi(e) \mid e \in S'_\pi\}$.
  Recall that $e_\pi \in S'_\pi$ by construction.
  Hence if $e_{\min}\neq e_\pi $, then $\pi(e_{\min})<\pi(e_\pi )$, which then implies $e_{\min}$ is not affected by $e_\pi $ in the original permutation $\pi$ and thus $e_{\min}\notin S_\pi$.

  Now suppose by way of contradiction that $e_{\min}\notin M_\pi$, then $e_{\min}$ has a neighboring edge $e'$ in $M_\pi$ such that $\pi(e')<\pi(e)$ since $M_\pi$ is maximal and was constructed greedily.
  Due to the minimality of $\pi(e_{\min})$, we must also have $e'\notin S'_\pi$, in which case $e_{\min}$ would not have been added to $S'_\pi$ at any step in the recursion, contradicting the definition of $e_{\min}$.
  Thus it follows that $e_\pi \in S'_\pi$.

  Due to the minimality of $\pi(e_{\min})$, it must be that $e_{\min} \in S'_{\pi,1}$, which implies that $e_{\min}$ intersects with $e_\pi $.
  But since $\pi(e_{\min})<\pi(e_\pi )$ and $e_{\min}$ intersects with $e_\pi $, then $e_\pi $ was not in $M_\pi$, and hence $S_\pi=\emptyset$.

  Suppose that $\pi(e_\pi )=\min \{\pi(e) \mid e \in S'_\pi\}$.
  We have nothing to show when $S_{\pi,0}=\emptyset$ because then $S_\pi =\emptyset\subseteq S'_\pi$.
  Suppose $S_{\pi,0}=\{e_\pi \}$.
  Then for each edge $e\in S'_{\pi,1}$, we have $\pi(e_\pi )<\pi(e)$ and thus $e\in S_{\pi,1}$.
  Moreover, each edge $e\notin S'_{\pi,1}$ has some neighboring edge $e'\in M_\pi$ such that $\pi(e')<\pi(e)$ and thus $e\notin S_{\pi,1}$.
  Hence, $S'_{\pi,1}=S_{\pi,1}$ and by induction, we have $S'_\pi=S_\pi$.
\end{proof}

For a permutation $\tau$ over $E$, we define $S'(\tau)=S'(G,G',\tau,e_\pi )$ as the set corresponding to $S'$ through the order of the edges induced by $\tau$.
We denote by $\Pi_F$ the set of all permutations $\tau$ for which it holds that $S'(\tau) = F$.

The proofs of Claims 4 and 5 in~\cite{censor2016optimal} can be directly used to show the following claims.
Again, the only difference is that we track permutations over edges rather than vertices.
Nevertheless, we give the proofs for completeness.
\begin{claim}\label{claim:mm:state}
  Let $F \subseteq E$ be a set of edges, and let $\pi$ and $\sigma$ be two permutations such that $\pi|_F = \sigma|_F$ and $\pi|_{E \setminus F} = \sigma|_{E \setminus F}$.
  Assume $\pi \in \Pi_F$.
  We have that $E \setminus F \subseteq E \setminus S'(\sigma)$ and every $e \in E \setminus F$ has the same state, i.e., whether or not $E\in M$, according to $\pi$ and $\sigma$.
\end{claim}
\begin{proof}
For $e\in E\setminus F$, we show that $e\in E\setminus S'(\sigma)$.
Moreover, we prove by induction that the order of the edges in $E\setminus F$ induces the same state under $\pi$ and under $\sigma$.

We first consider the base case, where $e\in E\setminus F$ has the smallest order, according to $\pi$ and $\sigma$.
Suppose by way of contradiction, that $e$ intersects with some edge $e'\in F$.
Since $e'\in F$ and $\pi\in \Pi_F$, then either before or after the graph update, we have $e'\in M_{\pi}$.
But $e\notin F$, so then $e$ cannot be in $M_{\pi}$, which implies the existence of some edge in $I_{\pi}(e)\cap E\setminus F$ that is in $M_{\pi}$.
However, this contradicts the minimality of $e\in E\setminus F$.
Hence, all neighbors of $e$ are in $E\setminus F$.

Because $\pi|_{E\setminus F}=\sigma|_{E\setminus F}$, then $e$ has smaller rank than all of its neighbors, according to $\sigma$.
Since $e$ also has smaller rank than all of its neighbors, according to $\pi$, then the matching $M$ has the same state upon the edge $e$ in both $\pi$ and $\sigma$.
Thus, $e\notin S'(\sigma)$ because $e\notin S'(\pi)$, which completes the base case.

To show the inductive step, consider a fixed edge $e\in E\setminus F$ and suppose the statement holds for all edges $e'\in E\setminus F\cap I_{\pi}(e)$.
We separate the analysis into cases, depending on whether $e$ is incident to any edges in $F$.

If $e$ is not incident to any edges in $F$, then either $e\in M_{\pi}$ or $e\notin M_{\pi}$.
First suppose $e\notin M_{\pi}$, so that some edge $z\in I_{\pi}\cap E\setminus F$ is in the matching $M_{\pi}$ induced by $\pi$.
Since $z\in E\setminus S'(\sigma)$ by the inductive hypothesis, then $z$ is also in the matching $M_{\sigma}$ induced by $\sigma$.
Because $\pi|_{E\setminus F}=\sigma|_{E\setminus F}$, then $e\notin M_{\sigma}$ as well so that $e$ has the same state according to $\pi$ and $\sigma$.
Similarly, if $e\in M_{\pi}$, then $w\notin M_{\pi}$ for any $w\in I_{\pi}(e)$.
Moreover, for any $w\in I_{\sigma}(e)$, we have by assumption that as a neighbor of $e$, $w\notin F$.
Thus $I_{\sigma}(e)\subseteq I_{\pi}(e)$ since $\pi_{E\setminus F}=\sigma_{E\setminus F}$.
By the inductive hypothesis, $w\in E\setminus S'(\sigma)$, so that $w\in M_{\sigma}$.
Hence, $e\in E\setminus S'(\sigma)$ and $e\in M_{\pi}$, so that its state is the same under $\pi$ and $\sigma$, as desired.

On the other hand, if $e$ is incident to some $w\in F$, then either before or after the update, we have $w\in M_{\pi}$, since $\pi\in\Pi_F$.
But $e\notin F$, so then $e\notin M_{\pi}$ and thus there exists $z\in I_{\pi}(e)\setminus F$ with $z\in M_{\pi}$.
By the inductive hypothesis, we have $z\in E\setminus S'(\sigma)$ and $z\in M_{\sigma}$.
Thus $\pi|_{E\setminus F}=\sigma|_{E\setminus F}$ implies $e\notin M_{\sigma}$ and $e\in E\setminus S'(\sigma)$, which completes the induction.
\end{proof}

\begin{claim}\label{claim:mm:subset}
  Let $F \subseteq E$ be a set of edges, and let $\pi$ and $\sigma$ be two permutations such that $\pi|_F = \sigma|_F$ and $\pi|_{V \setminus F} = \sigma|_{V \setminus F}$.
  Assume $\pi\in \Pi_F$.
  We have that $F\subseteq S'(\sigma)$.
\end{claim}
\begin{proof}
Let $e_\pi \in F$ be some fixed edge.
We use a similar strategy to show by induction on the order of the edges in $F$ according to $\pi$, with the modification that $e_\pi $ is the first edge in the permutation, that for each edge $e\in F$, we also have $e\in S'(\sigma)$.
We first consider $e_\pi $ as the base case.

Since $e_\pi \in F$ and $\pi\in\Pi_F$, then $e_\pi \in S'(\pi)$ and similarly, $e_\pi \in S'(\sigma)$.
Now for the inductive step, let $e\in F$ be an edge such that the statement holds for all edges in $F$ with smaller rank than $e$, according to $\pi$.
Because $e\in F$ and $e\neq e_\pi $, then there exists $w\in I_{\pi}(e)\cap F$.
Since $\pi|_F=\sigma|_F$ and $e\in F$, then by the inductive hypothesis, we have that $w\in S'(\sigma)$ and in particular $I_{\sigma}(e)\cap S'(\sigma)\neq\emptyset$.

Let $\phi\in I_{\sigma}(e)$, which is non-empty since $w\in I_{\pi}(e)$ and $\pi|_F=\sigma|_F$.
Now if $\phi\in F$, then since $e\in F$, we must have from our inductive hypothesis that $\phi\in S'(\sigma)$. 
If $\phi\notin F$, then $\phi\notin M_{\pi}$ in order for $e\in F$.
By Claim~\ref{claim:mm:state}, we thus have $\phi\in E\setminus S'(\sigma)$ and $\phi\notin M_{\sigma}$.
Hence, all neighbors of $e$ in $I_{\sigma}(e)$ are either in $S'(\sigma)$ or not in $M_{\sigma}$.
Since $I_{\sigma}\cap S'(\sigma)\neq\emptyset$, then $e\in S'(\sigma)$.
\end{proof}

These claims combined imply that if $\pi|_F = \sigma|_F$ and $\pi|_{V\setminus F} = \sigma|_{V\setminus F}$ then $\sigma \in \Pi_F$ if and only if $\pi \in \Pi_F$.

The following proof of Lemma~\ref{lem:prob:min} is almost exactly the same as that of Lemma 3 in~\cite{censor2016optimal}, with the focus on edges in a maximal matching rather than vertices in a maximal independent set.
\begin{lemma}\label{lem:prob:min}
  For any set of edges $F \subseteq E$, it holds that
  \[
    \Pr\Bigl[\pi(e_\pi) = \min\{\pi(e) \mid e \in F\} \mid S' = F\Bigr] = \frac{1}{|F|}.
  \]
\end{lemma}
\begin{proof}
Let $\tau$ be a fixed permutation.
Let $\sigma^+$ be a permutation on $F\setminus\{e_\pi\}$ and $\sigma^-$ be a permutation on $E\setminus F$.
Let $p_{\sigma^+,\sigma^-}:=\Pr[\pi(e_\pi)\le\pi(e)\; \forall \phi\in F\mid\pi|_{F\setminus\{e_\pi\}}=\sigma^+\,\wedge\,\pi|_{E\setminus F}=\sigma^-]$ denote the probability that $e_\pi$ has smaller rank than all edges $e \in F$ induced by a permutation $\pi$ that preserves $\sigma^+$ and $\sigma^-$.

We first claim that for any pair of permutations $\sigma^+_1,\sigma^-_1$ and $\sigma^+_2,\sigma^-_2$ on $F\setminus\{e_\pi\}$ and $E\setminus F$, respectively, then the permutation ${(\sigma^+_1)}^{-1}\sigma^+_2$ on $F\setminus\{e_\pi\}$ and the permutation ${(\sigma^-_1)}^{-1}\sigma^-_2$ on $E\setminus F$ are invariant on the property $\pi(e_\pi)\le\pi(e)$ for all $e \in F$.
Thus, $p_{\sigma^+_1,\sigma^-_1}=p_{\sigma^+_2,\sigma^-_2}$.
Since $e_\pi\in F$, then we have $\Pr[\pi(e_\pi)\le\pi(e)\,\forall e \in F]=\frac{1}{|F|}$ and hence,
\begin{align*}
\frac{1}{|F|}&=\Pr[\pi(e_\pi)\le\pi(e)\,\forall e\in F]=\sum_{\tau^+,\tau^-}p_{\tau^+,\tau^-}\Pr\left[\pi|_{F\setminus\{e_\pi\}}=\tau^+\,\wedge\pi|_{E\setminus F}=\tau^-\right]\\
&=\sum_{\tau^+,\tau^-}p_{\sigma^+,\sigma^-}\Pr\left[\pi|_{F\setminus\{e_\pi\}}=\tau^+\,\wedge\pi|_{E\setminus F}=\tau^-\right]=p_{\sigma^+,\sigma^-}.
\end{align*}
By Claim~\ref{claim:mm:state} and Claim~\ref{claim:mm:subset}, there exists a set of $t$ pairs of permutations $\{(\sigma^+_1,\sigma^-_1),\ldots,(\sigma^+_t,\sigma^-_t)\}$ on $F\setminus\{e_\pi\}$ and $E\setminus F$, respectively, such that $\Pi_F=\{\pi\mid\exists i, \pi|_{F\setminus\{e_\pi\}}=\sigma^+_i\wedge\pi|_{E\setminus F}=\sigma^-_i\}$ for every set $F\subseteq E$.
Thus,
\begin{align*}
\Pr[\pi(e_\pi)&\le\pi(e)\,\forall e\in F]=\sum_{i=1}^t p_{\sigma^+_i,\sigma^-_i}\Pr\left[\pi|_{F\setminus\{e_\pi\}}=\sigma^+_i\wedge \pi|_{E\setminus F}=\sigma^-_i \mid \pi\in\Pi_F\right]\\
&= \frac{1}{|F|}\sum_{i=1}^t\Pr\left[\pi|_{F\setminus\{e_\pi\}}=\sigma^+_i\wedge \pi|_{E\setminus F}=\sigma^-_i \mid \pi\in\Pi_F\right]=\frac{1}{|F|}.
\end{align*}
In other words, $\Pr[\pi(e_\pi) = \min\{\pi(e) \mid e \in F\} \mid S' = F] =\frac{1}{|F|}$, since  $\Pi_F$ is the set of all permutations $\tau$ for which it holds that $S'(\tau) = F$.
\end{proof}

The proof of Lemma~\ref{lem:expected:change} follows from Lemma~\ref{lem:possible:changes} and Lemma~\ref{lem:prob:min}.

\end{document}